\documentclass{article}
\usepackage[utf8]{inputenc}
\usepackage[T1]{fontenc}
\usepackage[english]{babel}

\usepackage{subcaption}
\usepackage{amsmath}
\usepackage{amssymb}
\usepackage{amsthm}
\usepackage{enumerate}
\usepackage{appendix}
\usepackage{geometry}
\usepackage{eurosym}
\usepackage{hyperref}
\hypersetup{
	colorlinks,
	citecolor=blue,
	filecolor=blue,
	linkcolor=blue,
	urlcolor=blue
}
\usepackage{todonotes}
\usepackage{mathrsfs}
\usepackage{dsfont}

\usepackage{rotating}  
\usepackage[round]{natbib}
\usepackage{empheq}
\usepackage{url}

\usepackage{booktabs}
\usepackage{adjustbox}

\def\be{\begin{eqnarray}}
	\def\ee{\end{eqnarray}}

\def\b*{\begin{eqnarray*}}
	\def\e*{\end{eqnarray*}}

\newtheorem{Theorem}{Theorem}[part]
\newtheorem{Definition}{Definition}[part]
\newtheorem{Proposition}{Proposition}[part]

\newtheorem{Remark}{Remark}[part]

\newcommand{\ba}{\begin{array}}
	\newcommand{\ea}{\end{array}}
\newcommand{\ben}{\begin{equation*}} 
	\newcommand{\een}{\end{equation*}}
\newcommand{\bea}{\begin{eqnarray}}
	\newcommand{\eea}{\end{eqnarray}}
\newcommand{\bean}{\begin{eqnarray*}} 
	\newcommand{\eean}{\end{eqnarray*}}
\newcommand{\bel}{\begin{align}} 
	\newcommand{\eel}{\end{align}}
\newcommand{\beln}{\begin{align*}} 
	\newcommand{\eeln}{\end{align*}}
\newcommand{\bit}{\begin{itemize}}
	\newcommand{\eit}{\end{itemize}}

\makeatletter \@addtoreset{equation}{section}

\@addtoreset{Definition}{section}

\@addtoreset{Theorem}{section}

\@addtoreset{Proposition}{section}

\@addtoreset{Property}{section}

\@addtoreset{Assumption}{section}

\@addtoreset{Corollary}{section}

\@addtoreset{Lemma}{section}

\@addtoreset{Remark}{section}

\@addtoreset{Example}{section}

\@addtoreset{Condition}{section}

\newcommand{\hs}{\hspace{3mm}}

\def \E{\mathbb{E}}

\def \H{\mathbb{H}}
\def \L{\mathbb{L}}
\def \M{\mathbb{M}}
\def \N{\mathbb{N}}
\def \P{\mathbb{P}}
\def \Q{\mathbb{Q}}
\def \R{\mathbb{R}}

\def \Z{\mathbb{Z}}

\def \G{\mathbb{G}}

\def\Tc{{\cal T}}

\def\={\;=\;}
\def\.{\;.}

\def\eps{\varepsilon}

\def\1{{\bf 1}}

\def\eps{\epsilon}

\def\normeL2#1{\left\|{#1}\right\|_{L^2}}

\newcommand{\alias}[2]{
	\providecommand{#1}{}
	\renewcommand{#1}{#2}
}

\alias{\P}{\mathbb{P}}
\alias{\N}{\mathcal{N}}
\alias{\L}{\mathcal{L}}
\alias{\Z}{\mathbb{Z}}
\alias{\Q}{\mathbb{Q}}
\alias{\R}{\mathbb{R}}
\alias{\C}{\mathcal{C}}
\alias{\T}{\mathbb{T}}
\alias{\E}{\mathbb{E}}
\alias{\H}{\mathcal{H}}
\alias{\1}{\mathbf{1}}
\alias{\B}{\mathcal{B}}
\alias{\M}{\mathcal{M}}
\alias{\G}{\mathcal{G}}
\alias{\Y}{Y_{\bullet}}

\newcommand{\nc}{\newcommand}
\nc{\cA}{{\mathcal A}} \nc{\cB}{{\mathcal B}} \nc{\cC}{{\mathcal
		C}} \nc{\cD}{{\mathcal D}} \nc{\bbD}{\mathbb{D}}
\nc{\cG}{{\mathcal G}} \nc{\cF}{{\mathcal F}} \nc{\cS}{{\mathcal
		S}} \nc{\cU}{{\mathcal U}} \nc{\cH}{{\mathcal H}}
\nc{\cK}{{\mathcal K}} \nc{\cM}{{\mathcal M}} \nc{\cO}{{\mathcal
		O}} \nc{\cP}{{\mathcal P}} \nc{\bbE}{\mathbb{E}}
\nc{\bbEP}{\mathbb{E}_{\mathbb{P}}}\nc{\bbL}{\mathbb{L}}
\nc{\bbP}{\mathbb{P}} \nc{\bbQ}{\mathbb{Q}} \nc{\del}{\partial}
\nc{\Om}{\Omega} \nc{\om}{\omega} \nc{\bbR}{\mathbb{R}}
\nc{\bbC}{\mathbb{C}} \nc{\bfr}{\begin{flushright}}
	\nc{\efr}{\end{flushright}} \nc{\dXt}{\delta q_{t}}
\nc{\dXs}{\delta q_{s}} \nc{\bs}{\blacksquare} \nc{\dX}{\delta q}
\nc{\dY}{\Delta Y}
\nc{\dnkx}{\left(X(T^{n}_{k})-X(T^{n}_{k-1})\right)}
\nc{\esssup}{\mathrm{ess}\mbox{ }\mathrm{sup}}
\nc{\essinf}{\mathrm{ess}\mbox{ } \mathrm{inf}}
\nc{\dhats}{\widehat{\delta_s}}

\nc{\chf}{\mbox{$\mathbf1$}}
\nc{\ind}{\mathds{1}}

\nc{\htau}{\hat\tau}
\nc{\btau}{\bar\tau}

\begingroup\expandafter\expandafter\expandafter\endgroup
\expandafter\ifx\csname pdfsuppresswarningpagegroup\endcsname\relax
\else
\fi
\graphicspath{{./Data/}{./figs/}}

\newcommand{\dd}{\mathop{}\mathopen{}\mathrm{d}}

\begin{document}
\title{Exit Incentives for Carbon Emissive Firms}
\author{René Aïd\thanks{Université Paris Dauphine - PSL Research University.}
\qquad
Xiangying Pang\thanks{The Chinese University of Hong Kong, Department of Mathematics.}
\qquad
Xiaolu Tan\thanks{The Chinese University of Hong Kong, Department of Mathematics, research supported by Hong Kong RGC General Research Fund (project 14302622) and the Faculty Direct Grant.}
}
\maketitle
	
\begin{abstract}
We develop a continuous-time model of incentives for carbon emissive firms to exit the market based on a compensation payment identical to all firms. In our model, firms enjoy profits from production modeled as a simple geometric Brownian motion and do not bear any environmental damage from production. A regulator maximises the expected discounted value of firms profits from production minus environmental damages caused by production and proposes a compensation payment whose dynamics is known to the firms. We provide in both situations closed-form expressions for the compensation payment process and the exit thresholds of each firms.  We apply our model to the crude oil market. We show that market concentration both reduces the total expected discounted payment to firms and the expected closing time of polluting assets. We extend this framework to the case of two countries each regulating its own market. The presence of a second mover advantage leads to  the possibility of multiple equilibria. Applying this result to large producing countries, we find that they are unlikely to agree on the timing to exit market.
\end{abstract}

\hs

{\bf keywords}: incentives, auction, stochastic game, stopping times, pollution control.

\hs


\hs

{\bf AMS}:  60G40, 91A15,  91B03, 91B26, 91B41, 93E20.

\tableofcontents

\section{Introduction}
	
Since the signature of Kyoto's protocol in 1997, a lot of efforts have been engaged by countries around the world to reduce their green house gas emissions rate. To give an idea of this effort, it suffices to look at the evolution of the renewable energy installed capacities in the European Union between 2000 and 2024. In 2000, there were respectively 13~GW and 0,6~GW of installed capacity of wind and solar production while these numbers are 265~GW and 338~GW in 2024. Even if the investment cost in these two technologies dropped significantly due to economies of scale and technological progress, at a rough estimate of 1~billion euro/GW, this makes a 500~billion euro spending over the period. Compared to this 40 times fold multiplication of non-emissive installed generation capacity, one should also observe that the installed capacity of coal-fired plant only dropped by 33\% over the same period, going from 150~GW to 100~GW. This illustrates a major issue in the fighting of green house gaz emission: it is easier to add green production than to remove brown assets. In the former case, one provides new income to the investor while in the latter, it requires to suppress an income stream to the owner.\footnote{The figures appearing in this paragraph are taken from International Renewable Energy Agency publication for costs (Renewable Power Generation Costs in 2023), from Wikipedia's web page on European Wind production.}\textsuperscript{,}\footnote{Of course, relative to the European Union GDP of $17\,000$ billion euro this investment of 25~billion euro per year is a drop of water. It is also a drop of water relative to EU yearly public spending of $7800$ billion euros.}

This work is precisely concerned on the difficult part of the decarbonation of the global economy which is the removal of emissive assets that are currently providing profits to firms. In environmental economics, this problem is usually tackled either using a taxation of the polluters with the so-called Pigouvian tax (\cite{Pigou20}, \cite{Boyer99}) or a cap-and-trade system finding its origin in the seminal work of \cite{Montgomery72}. If the efficiency of these economic mechanisms is recognized to induce a pollution reduction (see \cite{Kohn94} for a theoretical analysis and \cite{Eslahi24} for a recent assessment of the efficiency of the European Union cap-and-trade system of carbon emissions), it is not clear wether it is enough to induce a market exit of the polluting firms (see \cite{Biorn98} for an empirical analysis). In this paper, we consider an alternative incentive mechanism based on the compensation of firms for their value lost because of market exit. 

Our framework is based on the incentive auction mechanism developed in \cite{He23},  which is based on the Bank-El Karoui's representation of stochastic processes in \cite{BEK2004}. In our model, firms enjoy profits from production modelled as a simple geometric Brownian motion and do not bear any environmental damage from production. Parameters are set in such a way that firms have no incentive to ever quit the market. A regulator maximises the expected discounted value of firms profits from production minus environmental damages caused by production. At each instant, the regulator proposes a compensation payment whose dynamics is known to the firms. Knowing the payment dynamics, each firm can decide whether the current payment value compensate its market exit. The design of the compensation process is made difficult by the fact that {\em any} firm, large or small, can be a claimant of the proposed compensation. 

Our model allows to answer to several simple research questions on the decarbonation of economy. First, it allows to assess how much does it cost to simply compensate polluting firms to stop their activities and how much time does it take. Indeed, in our model the regulator uses discounting to maximise a compromise between the profits from production which provides a service to society and the environmental damages. Thus, as long as present profits outweigh future damages, it is socially optimal for the regulator to keep production activity. We compare the second-best incentive mechanism induced by our compensation mechanism to a first-best situation where the regulator would directly make the compromise between profits and damages and decide when to stop. Our model also allows to asses how market fragmentation or concentration affects the cost and timing of market exit. Moreover, we consider a situation where production activity is under the control of two regulators. This case captures the conflict between two countries who enjoys individual profit from production but collective environmental damages. In this last situation, regulators optimal incentives compensations form a Nash equilibrium. In both cases, one or two regulators, we illustrate our findings using the crude oil market. If one would be willing to reach the emission targets of the Paris agreement of limiting global warming to 1.5 degree, one had not only to stop immediately oil consumption but also to extract carbon dioxyde from the atmosphere. But,  oil industry represents a business of 100~million barils per day, which amounts at roughly 1\% of world GDP on its own and 1.4 billions road vehicles in the world, which makes not an easy task to stop immediately oil production.

\hs

\noindent{\bf Results} First, in the case of a single regulator of a market served with $N$ firms, we provide a closed-form expression of an optimal incentive compensation process and the induced optimal exit times of each firm. In the case of a game of stopping times between two regulator managing their market served by a monopolistic firm, we characterise the different Nash equilibria that may arise and give closed-form expressions for the exit thresholds in all situations. In particular, we show that under mild conditions, there are two possible Nash equilibria. Indeed, the situation faced by the two regulators involves a second mover advantage: the second to exit benefits from less damages induced by the disappearance of the other production and still enjoys its own profit from production.  Second, using the crude oil market as an illustrative device, we analyses the effect of the discount rate, of the number of firms and of the market concentration on the expected payment from the regulator and the expected time to halt all production. We observe that large discount rate postpone the moment to cease production but reduces the amount to be paid to terminate it. The number of firms as well as market concentration have an increasing effect on the total expected payment required to halt all production. But, market concentration has a much larger effect on this variable. Moreover,  while the number of firms increases the expected time to close production, market concentration tends to drastically reduces it. Hence, it is more costly but quicker to close a highly concentrated market rather than a fragmented market.

\hs

\noindent{\bf Related literature} The present work lies at the intersection between optimal incentive theory in continuous-time,  auction theory and optimal stopping with applications to decarbonation of the economy. The incentive mechanism proposed here owes to \cite{Sannikov08} seminal paper on Principal-Agent in continuous-time and \cite{Cvitanic18} for its generalisation using second-order backward stochastic differential equations. In terms of economic applications, a closely related model of decarbonation of electricity generation using precisely this framework has been recently developed in \cite{Aid25} where carbon emissive production capacities are removed at an absolutely continuous rate. The same problem is also investigated in \cite{Kharroubi19} using a similar framework. Models of decarbonation using explicitly market exit modelled by stopping times includes \cite{Aid21} using a mean-field approach and numerical resolution of the related PDEs and also \cite{Bouveret22}. Less related but close enough to be cited, \cite{Nutz22}'s paper develops an incentive mechanisms to make a crowd of agents exit a game depending on their ranking. Our paper also owes to the literature of auction theory (see \cite{Milgrom04} and \cite{Krishna09} for introductions). Indeed, our compensation process is the promise of payment and the regulator can be interpreted as an auctioneer.

\hs

\noindent{\bf Roadmap} The rest of the paper is organized as follows. In Section \ref{sec:one-country}, we introduce our model of a single market with one regulator and multiple carbon emissive firms. Under the geometric Brownian motion production framework, we provide an explicit solution to the regulator's optimal exit incentives problem using hitting times as exit time. Then in Section \ref{sec:2markets}, we consider a model with two interactive markets with two different regulators and study a Nash equilibrium problem. We also compare the Nash equilibrium solution to the first best and second best solutions when the two market have the same regulator. The explicit solutions are then implemented in Section \ref{sec:eco_application} with real world data. Finally, some technical proofs are reported in Appendix.

\section{The case of a single market}
\label{sec:one-country}

\subsection{Model}
\label{subsec:model}

We consider that the market for a carbon-intensive good like crude oil or coal is served by $N \geq 1$ producers. The total production rate  $X_t$ is governed by a geometric Brownian motion with volatility $\sigma >0$ and an average growth rate $\mu >  \sigma^2/2$ according to the dynamics 
$$
	\dd X_t = \mu X_t \dd t + \sigma X_t \dd W_t, 
	~~~
	X_0 = x_0,
$$ 
where $W_t$ is a standard Brownian motion, $x_0 > 0$ is the initial production intensity. Since our focus is on long-term decision-making, we disregard the short-term fluctuations in profits among the firms serving the market and assume that these profits are directly proportional to the level of production. These firms can be classified based on average market share $\lambda_i>0$ summing to 1 and unit profit $\pi_i>0$ increasing in $i$, i.e.
	\begin{equation} \label{eq:lambda_pi_condition}
		\sum_{i=1}^N \lambda_i = 1,
		~~\mbox{and}~
		\pi_i < \pi_{i+1}, ~~i = 1, \cdots, N-1.
	\end{equation}

	Then then expected discounted profit at rate $\rho$ of firm $i$  over the period $(0,T)$ is given by
	\begin{align*}
		\E\Big[\! \int_0^T e^{-\rho t} \pi_i \lambda_i X_t \dd t \Big] 
		~=~
		\frac{\pi_i \lambda_i X_0}{\rho - \mu} \big(1 - e^{-(\rho-\mu)T} \big).
	\end{align*}
	Thus, the firm has no incentive ever to stop. Further, we assume $\rho>\mu$, so that the total expected discounted profit of firm $i$ is
	\begin{equation} \label{eq:def_Z0}
		Z^i_0 
		~:=~
		\E\Big[\!  \int_0^{+\infty} e^{-\rho t} \pi_i \lambda_i X_t \dd t \Big] 
		~=~
		\frac{\pi_i \lambda_i X_0}{\rho - \mu}.
	\end{equation}
	
	The carbon-intensive production induces damages, which is not borne by the firms. 
	We assume that the damage function is of the form  
	$$
		L(x) := \ell x^{\gamma},
		~~\mbox{with constants}~\ell > 0~\mbox{and}~\gamma \geq 2.
	$$ 
	Thus, the total social surplus of production (profit from firms minus damages) over a period $(0,T)$ is given by
	\begin{align*}
		U_T ~:=~&
		\E\Big[ \int_0^{T} e^{-\rho t} \Big(\bar \pi_1 X_t   - \ell X_t^{\gamma} \Big)\dd t\Big] \nonumber \\
		=~&
		\frac{ \bar \pi_1 X_0}{\rho - \mu} (1- e^{-(\rho-\mu)T}) + \ell X_0^{\gamma} \frac{e^{-(\rho -\gamma \mu- \frac{1}{2}\sigma^2(\gamma^2 - \gamma)) T }-1}{\rho -\gamma \mu - \frac{1}{2}\sigma^2(\gamma^2 - \gamma)},
		~~\mbox{with}~
		\bar \pi_1 := \sum_{i=1}^N \pi_i \lambda_i.
	\end{align*}
	When $\rho -\gamma \mu- \frac{1}{2}\sigma^2(\gamma^2 - \gamma)<0$, the problem becomes interesting since $U_T \to -\infty$ as $ T \to +\infty$. 
	Hence, we assume that 
	$$
		\rho \in (\mu,\gamma \mu+ \frac{1}{2}\sigma^2(\gamma^2 - \gamma)).
	$$

\hs

	We look for an appropriate incentive mechanism that encourages the  firms to exit the market. 
	Let $\mathbb{F}$ be the natural filtration generated by $X_t$,
	we denote by $\mathcal{T}$ the set of all $\mathbb{F}-$stopping times such that 
	$$
		\E \big[ e^{-\rho \tau}X_{\tau}^{\gamma} \log^+(e^{-\rho \tau} X_{\tau}^{\gamma}) \big] <\infty,
		~~\mbox{with}~
		\log^+(x) = \max(\log (x),0).
	$$
	The regulator seeks an incentive mechanism $Y = (Y_t)_{t \ge 0}$ that is adapted to the production history $X = (X_t)_{t \ge 0}$ such that each firm $i$ solves 
	\begin{align} \label{eq:pb_agent_i}
		\sup_{\tau_i \in \mathcal{T}} \E\Big[ \int_0^{\tau_i} e^{-\rho t} \pi_i  \lambda_i X_t \dd t + \lambda_i e^{-\rho \tau_i}   Y_{\tau_i} \Big],
	\end{align}
	and the regulator wants to optimize
	\begin{align} \label{eq:pb_principal}
	v(x_0) := \sup_{Y} \E \bigg [ \!\int_0^{\infty} \!\!\! e^{-\rho t} \bigg( \sum_{i = 1}^N \pi_i \lambda_i  X_t \1_{\{t \le \htau_i\}} - \ell \Big(\sum_{i=1}^N \lambda_i X_t \1_{\{t \le \htau_i\}} \Big)^{\gamma}\bigg) \dd t - \sum_{i=1}^N \lambda_i e^{-\rho \htau_i} Y_{\htau_i} \bigg| X_0 = x_0 \bigg],
	\end{align}
where $\htau_i$ is the best-response of the firm $i$ to the incentive scheme $Y$. 
	
\hs

By a trivial extension of the work of \cite{He23} from the uniform bounded stopping times to the stopping times in $\mathcal{T}$, the optimation problem of the regulator can be brought back to a multiple stopping problem. Indeed, since $X_t > 0$ and $\pi_i < \pi_j$ for $i < j$, it follows that, with the same exit incentive scheme $Y$, the optimal stopping times $(\htau_i)_{i=1, \cdots, N}$ satisfies $\htau_i < \htau_{i+1}$.
	Therefore, the regulator's problem in \eqref{eq:pb_principal} is equivalent to
	\begin{align} \label{eq:regu}
		\sup_{Y}~
		\E \bigg[ 
		\sum_{i=1}^N  \Big(
		\int_{\htau_{i-1}}^{\htau_i} e^{-\rho t} \Big( \sum_{j = i}^N \pi_j \lambda_j  X_t - \ell \Big( \sum_{j=i}^N\lambda_j X_t \Big)^{\gamma} \Big) \dd t - \lambda_i e^{-\rho \htau_i} Y_{\htau_i} 
		\Big)
		\bigg| X_0 = x_0 \bigg],
	\end{align}
	with $\htau_0 := 0$. 
	Further, following \cite{He23},  the regulator's problem \eqref{eq:pb_principal} is further equivalent to the multiple stopping problem:
	\begin{align} \label{eq:regu_equiv}
		v(x_0) = \sup_{\tau_1 \leq \cdots \leq \tau_N} \E\Big[
		\sum_{i=1}^N \Big( 
		\int_{\tau_{i-1}}^{\tau_i} e^{-\rho t} \big( \bar \pi_i X_t - \ell (1-\tilde \lambda_{i-1})^{\gamma} X_t^{\gamma} \big) \dd t  
		-
		e^{-\rho \tau_i} (\tilde{\lambda}_i c_i - \tilde{\lambda}_{i-1} & c_{i-1}) X_{\tau_i} 
		\Big) \nonumber \\
		&~~
		\Big| X_0 = x_0
		\Big],
	\end{align}
	with $\tau_0 := 0,  c_0 := 0,   \tilde{\lambda}_0 := 0,$ and for each $i=1, \cdots, N$,
\begin{equation}  \label{def:coef}
c_i = \frac{\pi_i }{\rho-\mu},\quad
\bar \pi_i := \sum_{j=i}^N \pi_j \lambda_j, \quad
\tilde{\lambda}_i := \sum_{j=1}^i \lambda_j. 
\end{equation}

\subsection{Main result}
\label{subsec:1market}

	Our first main result consists of a resolution of the principal's problem \eqref{eq:pb_principal}.
	Let
	\begin{align} \label{def:hitting_point}
		\hat x_{i} ~:=~ \Big( \frac{b(m-1)}{a \ell (\gamma-m)}  ~\frac{\pi_i(\lambda_i + \tilde{\lambda}_i) - \pi_{i-1} \tilde{\lambda}_{i-1} } { \big(1-\tilde{\lambda}_{i-1} \big)^{\gamma} - \big(1-\tilde{\lambda}_{i} \big)^{\gamma} } \Big)^{1/({\gamma}-1)},
		~~i = 1, \dots, N,
	\end{align}
	and
	\begin{equation} \label{eq:def_abm}
		a := \frac{1}{\gamma [\mu + \frac{1}{2} \sigma^2 (\gamma - 1)] -\rho} >0,
		~~
		b:= \frac{1}{\rho-\mu} >0,
		~~
		m ~:=~ \frac{1}{2} - \frac{\mu}{\sigma^2} + \sqrt{ \Big( \frac{1}{2} - \frac{\mu}{\sigma^2} \Big)^2 +\frac{2\rho}{\sigma^2} }.
	\end{equation}
	Recall that $\gamma \ge 2$ and $\rho \in (\mu,\gamma \mu+ \frac{1}{2}\sigma^2(\gamma^2 - \gamma))$, then one has $m\in (1,\gamma)$.

\begin{Theorem}[{\rm Optimal exit incentives for $N$ firms}]\label{theo:Nfirms} 
	Let $(\lambda_i)_{i=1,\cdots, N}$, $(\pi_i)_{i=1, \cdots, N}$ satisfy \eqref{eq:lambda_pi_condition}, $\gamma \ge 2$ and $\rho \in (\mu,\gamma \mu+ \frac{1}{2}\sigma^2(\gamma^2 - \gamma))$.
	Then the following hold true.
	
	\vspace{0.5em}
	
	\noindent $\mathrm{(i)}$ The constants $(\hat x_i)_{i=1, \cdots, N}$ defined by \eqref{def:hitting_point} satisfy ${\hat x_1} < \hat x_2 < \cdots < \hat x_N$.

\vspace{0.5em}
	
	\noindent $\mathrm{(ii)}$ Let $\eps>0$. An optimal incentive payment process $\widehat Y = (\widehat Y_t)_{t \ge 0}$ is given by

\begin{equation} \label{eq:def_Yhat}
	\widehat Y_t 
	~:=~
	\sum_{i=1}^{N} \Big( \big( c_i X_t + d_i X_t^m -\eps \big) \1_{\{\htau_{i-1} < t < \htau_i\}}  +  \big( c_i X_t + d_i X_t^m \big) \1_{\{t = \htau_i\}} \Big),
	~t \ge 0,
\end{equation}
	where $c_i$ is defined in \eqref{def:coef}, $\hat \tau_0 = 0$, $\hat \tau_i :=\inf\{ t \ge 0 ~: X_t \ge \hat x_i \}$ and
	\begin{align*}
	d_i := \sum_{j=i+1}^N b(\pi_j - \pi_{j-1}) \Big(\1_{ \{x_0 \geq \hat x_j\}}  x_0^{1-m} + \1_{ \{x_0 < \hat x_j\}} \hat x_j^{1-m} \Big),
	~~i = 1, \cdots, N.
	\end{align*}

\vspace{0.5em}
	
	\noindent $\mathrm{(iii)}$ Under the optimal incentive above, firms $i$ exits the market at time $\hat \tau_i$ for each $i=1, \cdots, N$.
		
\vspace{0.5em}
	
	\noindent $\mathrm{(iv)}$ For any initial condition $x_0 > 0$, the value function in \eqref{eq:pb_principal} is given by
\begin{align} \label{eq:value_fun}
	v(x_0)  = & \sum_{i = 1}^N \bigg(  \1_{ \{x_0 < \hat x_i\}} \big( -2 \pi_i \lambda_i   b \hat x_i  -  \tilde \lambda_{i-1}  (\pi_i - \pi_{i-1}) b  \hat x_i - a \ell \big((1-\tilde \lambda_{i-1})^\gamma - (1-\tilde \lambda_{i})^\gamma \big) \hat x_i^\gamma \big) \big( \frac{x_0}{\hat x_i}\big)^m \nonumber \\ 
	&~~~~~~~ + \1_{ \{ \hat x_{i-1} \leq x_0 < \hat x_i\}} \big( b \bar \pi_i x_0 + a \ell (1-\tilde \lambda_{i-1})^\gamma x_0^\gamma \big) - \1_{ \{x_0 \geq \hat x_i\}}(\pi_i \tilde \lambda_i - \pi_{i-1} \tilde \lambda_{i-1})bx_0 \bigg) ,
\end{align}
where $\hat x_0 := 0$ and $\tilde \lambda_i$, $\bar \pi_i$ are defined by \eqref{def:coef}.

\end{Theorem}

Notice that the optimal incentive payment process is not unique, one can always add to $\widehat Y$ in \eqref{eq:def_Yhat} a process $Y'$ such that $Y'_t \le 0$ for all $t \ge 0$ and $Y'_{\htau_i} = 0$ for each $i=1, \cdots, n$. Then the optimal exit time of firm $i$ will still be $\htau_i$ and the effective payment $\widehat Y_{\htau_i}  + Y'_{\htau_i}$ is the same as previous value $\widehat Y_{\htau_i} $. Otherwise, one can also add to $\widehat Y$ a process $M$ such that $e^{-\rho t} M_t$ is martingale with expectation value $0$, so that it does not change the optimal exit time of each firm and the expectation value of the effective payment.

The expression of $v(x_0)$ in \eqref{eq:value_fun} includes the expected revenue from the production, minus the social damage as well as subsidies. Let us provide the computation of the expected valued of the discounted subsidies below:
	
	\begin{itemize}
	\item[(i)] When $x_0 < \hat x_i$, one has $\htau_i >0$, $X_{\htau_i} = \hat x_i$, and $\E [ e^{-\rho \htau_i} \widehat{Y}_{\htau_i}]  = c_i \hat x_i^{1-m} x_0^m + d_i   x_0^m$ 
	
	\item[(ii)] When $x_0 \geq \hat x_i$, one has $\htau_i = 0$, $X_{\htau_i} = x_0 $ and $\E [ e^{-\rho \htau_i} \widehat{Y}_{\htau_i}]  = c_i x_0 + d_i x_0^m$.
	\end{itemize}
	The proof will be reported in Section \ref{app:main1} in Appendix.

	\vspace{0.5em}

	With the optimal incentive payment process $\widehat Y$, the firm $i$ exits the market at the hitting time $\hat \tau_i$ and receives the payment
	$$
		\widehat Y_{\hat \tau_i} 
		~=~
		c_i X_{\hat \tau_i} + d_i X^m_{\hat \tau_i} 
		~=~
		c_i \hat x_i + d_i \hat x_i^m.
	$$

	\noindent $\mathrm{(i)}$ When $N=1$, one can compute that $c_1 = \frac{\pi_1}{\rho-\mu}$ and $d_1 = 0$ and hence $\widehat Y_{\htau_1} = Z^1_{\htau_1}$ where $Z$ is the total expected discounted profit of firm $1$ without stopping, as defined in \eqref{eq:def_Z0}.
	Namely, to incitivise the firm to exit at time $\htau_1$, the regulator needs to provide the compensation $\widehat Y_{\htau_1} = Z^1_{\htau_1}$ which equivalent to Firm's total profit in the future if it never exits.

	\vspace{0.5em}

	\noindent $\mathrm{(ii)}$ When $N=2$, one can compute that 
	$$
		c_1 = \frac{\pi_1}{\rho-\mu},
		~~
		d_1 = \frac{\pi_2 - \pi_1}{\rho - \mu} \hat x_2^{1-m},
		~~
		c_2 = \frac{\pi_2}{\rho - \mu},
		~~
		d_2 = 0.
	$$
	For Firm $2$, the compensation at its exit time $\htau_2$ is $\widehat Y_{\htau_2} = Z^2_{\htau_2} = c_2 X_{\htau_2}$, which is equivalent to total future profit.
	However, for Firm $1$, the compensation at its exit time $\htau_1$ becomes $\widehat Y_{\htau_1} = c_1 X_{\htau_1} + d_1 X_{\htau_1}^m $ which is greater than its total future profit $Z^1_{\htau_1} = c_1 X_{\htau_1}$.
	In fact, due to the incitation payment $\widehat Y_{\htau_2} = c_2 X_{\htau_2}$ at time $\htau_2$, 
	Firm $1$ would have a different expected future profit.

Concretely, at time $\htau_1$, Firm $1$ has $c_1 X_{\htau_1} + d_1 X_{\htau_1}^m $ as optimal expected discounted profit with optimal exit time $\htau_1$.  This amount is the sum of the expected future profit of Firm 1 during the intervals $\htau_1$ and $\htau_2$, and the expected discounted value of $\widehat Y_{\htau_2}$.  We observe that, under the uniform compensation mechanism, this sum is greater than the value $Z^1_{\htau_1} = c_1 X_{\htau_1}$ which corresponds to the expected discounted profit of Firm 1 it it never exits.

The way firms are incentivised to leave the market is consistent with their size and their profit shares. Because each firm can seize the proposed amount $Y_t$, the regulator is compeled to design a nondecreasing payment process to make the payment attractive to small firms and non-interesting for large firms. This might look unfair as major firms are staying in the market for the longest period whereas they are the ones who pollute the most.

\subsection{The case of a single firm}

The case of a market served by a single firm allows to get insights in the incentive mechanism. Besides, it makes it easier to compare the properties of the incentive mechanism with a social optimum where the regulator would directly decide the stopping time that realises the optimal compromise between damage costs and production profit.

\hs

Applying Theorem~\ref{theo:Nfirms} with $N=1$ directly provides the optimal incentive payment $Y$ and the exit threshold $\hat x_1$ as
\begin{equation} \label{def:1Y}
	\widehat Y_t ~:=~   \frac{\pi_1}{\rho-\mu}\Big[ \big( X_{t} - \eps \big) \1_{ \{{t<\htau_1} \}} +  X_{t} \1_{ \{{t = \htau_1} \}}\Big], ~t \ge 0,
	~~\mbox{and}~
	\hat x_1 := \Big(\frac{m-1}{\gamma - m} \frac{b}{a} \frac{2 \pi}{\ell} \Big)^\frac{1}{\gamma-1},
\end{equation}
where
$$
	\hat \tau_1 :=\inf\{ t \ge 0 ~: X_t \ge \hat x_1 \}.
$$
	Moreover, the smallest best-response of the firm with incentive scheme $\widehat Y$ is also given by $\htau_1$, 
	then by Proposition \ref{prop:payment}, the expected compensation value is
	$$
		\E \big[e^{-\rho \htau} Z^1_{\htau_1} \big] ~=~ \frac{c x^m}{(\hat x_1)^{m-1}} ~=~ c x^m\Big(\frac{m-1}{\gamma - m} \frac{b}{a} \frac{2 \pi}{\ell} \Big)^\frac{1-m}{\gamma-1}, 
	$$
	and the total utility value is
	$$
		v(x_0) = \Big( b \pi x_0 + a\ell x_0^{\gamma}+ (-2  b \pi {\hat x_1} - a \ell {\hat x_1}^{\gamma}) \big( \frac{x_0}{{\hat x_1}} \big)^m \Big) \1_{\{ x_0<{\hat x_1}\}} - b  \pi  x_0 \1_{\{ x_0 \ge \hat x_1 \}}.
	$$

	\vspace{0.5em}

As a comparison, we consider the first-best situation where it would be possible to decide for the whole society the optimal moment when the damage costs outweighs the industry profit and it is time to quit.  In this scenario, one would address the standard optimal stopping problem given by
\begin{align} \label{eq:OT_1}
u(x) 
~:=~
\sup_{\tau \in \mathcal{T}} ~\E\Big[  \int_0^{\tau} e^{-\rho  t} \big( \pi X_t - \ell  X_t^{\gamma}\big) \dd t   \Big], \quad X_0 = x.
\end{align}
This is a standard stopping time  problem which solution is detailed in Section \ref{app:comp}. We have
\begin{align*}
	u(x_0) = \Big(  b \pi x_0+a \ell x^{\gamma}_0 + ( -b \pi \bar x_1-a \ell \bar x_1^{\gamma} )\Big( \frac{x_0}{\bar x_1} \Big)^m \Big) \1_{\{ x < \bar x_1 \}},
	~~\mbox{with}~~
	\bar x_1 :=  \Big( \frac{m-1}{\gamma - m} \frac{b}{a}  \frac{\pi}{\ell } \Big)^{\frac{1}{\gamma-1} }.
\end{align*}
Moreover, the hitting time $\btau_1 := \{t \ge 0~: X_t \ge \bar x_1 \}$ is an optimal solution to \eqref{eq:OT_1}.

\begin{Remark}{\rm 

	$\mathrm{(i)}$ According to the definition of $m$, we observe that $ \frac{m-1}{\gamma - m} \frac{b}{a} >1$. Thus, for all $\gamma \geq 2$, we have 
$$  \bar x_1 = \Big( \frac{m-1}{\gamma - m} \frac{b}{a}  \frac{\pi}{\ell } \Big)^{\frac{1}{\gamma-1} } > \Big(\frac{\pi}{\ell}\Big)^{\frac{1}{\gamma-1}}.$$	
Note that the ratio $(\frac{\pi}{\ell})^{\frac{1}{\gamma-1}}$ is the level after which the damage costs outweigh the profit, i.e. $f(x)<0$ when $(\frac{\pi}{\ell})^{\frac{1}{\gamma-1}} < x$. 

	\vspace{0.5em}
		
	\noindent $\mathrm{(ii)}$ We find that $\hat x_1 = 2^{\frac{1}{\gamma-1}} \bar x_1 > \bar x_1$ as $2^{\frac{1}{\gamma-1}} \in (1,2],$ which means more time is needed for the firm to exit the market under the incentive mechanism. That makes sense since the firm's revenue is sacrificed in the central planner case, and the regulator also needs to control subsidies in an affordable range. 
}
\end{Remark}

\section{The case of two interacting markets}\label{sec:two-countries}
\label{sec:2markets}

	We now consider the case with two interacting markets, where 	the social benefit from the production is individual, 
	and the damage from the carbon emission is global based on the total emission of the two countries.
	We study the problem in two different settings.
	In a first setting, there is one regulator for each market, and each regulator provides an incentive scheme to its own industry to stop the carbon intensive production.
	We look for the Nash equilibrium solutions of the problem.
	In a second setting, there is only one central planner.
	 The solution planner either applies the compensation scheme to incentivise the firm to stop the carbon intensive production and one looks for the second best solution,
	or decides the exit time of each of the two firms to stop and one looks for the first best solution.

\subsection{The setting with two regulators and Nash equilibrium}
\label{subsec:Nash}

	Let us first introduce the Nash equilibrium problem in the 2 regulator-firms setting.
	Given the exit times $\tau_1$ and $\tau_2$ of the two firms, the global social damage from the carbon emission is given by
	\begin{equation} \label{eq:def_Damage2}
	D(x_0, \tau_1, \tau_2)
	:= 
	\E \Big[ \int_0^{\infty}  e^{-\rho t} \ell \big(\lambda_1 X_t \1_{\{t \le \tau_1\}} + \lambda_2 X_t \1_{\{t \le \tau_2\}} \big) ^{\gamma}  dt \Big].
	\end{equation}
	This damage is half-half shared by the two countries. 
	In view of result in Section \ref{subsec:1market}, when Regulator~1 uses the compensation process to incentivise Firm 1 to exist at time $\tau_1$, he/she  needs to provide a compensation $Z^1_{\tau_1}$ with 
	$$Z_t^1:=\E \Big[ \int_t^{\infty} e^{-\rho(s-t)} \pi_1 \lambda_1 X_s ds \Big]  = \frac{\pi_1 \lambda_1}{\rho - \mu} X_t = \pi_1 \lambda_1 b X_t.$$ 
	Therefore, given the stopping time $\tau_2$ of the second firm, the optimal compensation problem of Regulator 1 becomes the following stopping problem:
	\begin{align}\label{eq:pb_regulator1}
		\sup_{\tau_1 \in \Tc} J_1(x_0,\tau_1,\tau_2) 
		~:=~
		\sup_{\tau_1 \in \Tc} \E \Big[ \int_0^{\tau_1}  e^{- \rho t}  \pi_1 \lambda_1 X_t dt - e^{-\rho \tau_1} Z_{\tau_1}^1 \Big] 
		- \frac12 D(x_0, \tau_1, \tau_2).
	\end{align}
	Similarly, given the stopping time $\tau_1$ of the first firm, Regulator 2 solves the other optimal stopping problem:
	\begin{align} \label{eq:pb_regulator2}
		\sup_{\tau_2 \in \Tc} J_2(x_0,\tau_1,\tau_2) 
		~:=~
		\sup_{\tau_2 \in \Tc} \E \Big[ \int_0^{\tau_2} e^{- \rho t} \pi_2 \lambda_2 X_t  dt - e^{-\rho \tau_2} Z_{\tau_2}^2 \Big] - \frac12 D(x_0, \tau_1, \tau_2).
	\end{align}
	where
	$$
		Z_t^2
		~:=~
		\E \Big[ \int_t^{\infty} e^{-\rho(s-t)} \pi_2 \lambda_2 X_s ds \Big]  = \frac{\pi_2 \lambda_2}{\rho - \mu} X_t = \pi_2 \lambda_2 b X_t.
	$$

	\begin{Definition}  \label{def:Nash}
		$\mathrm{(i)}$ A couple $(\tau^*_1, \tau^*_2)$ of stopping times is a Nash equilibrium of the 2 non-cooperative regulators problem 
		if 
		\begin{itemize}
			\item $\tau^*_1$ is solution to \eqref{eq:pb_regulator1} with given $\tau_2$, 
			\item $\tau^*_2$ is solution to \eqref{eq:pb_regulator2} with given $\tau_1$.
		\end{itemize}
		
		\noindent $\mathrm{(ii)}$ Given a Nash equilibrium $(\tau^*_1, \tau^*_2)$, 
		the corresponding total social utility is defined by
		\begin{align*}
			u(x_0)
			~:=~
			J_1(x_0, \tau^*_1, \tau^*_2) + J_2(x_0, \tau^*_1, \tau^*_2) .
		\end{align*}

	\end{Definition}

	Our main result is the existence of the Nash equilibrium which are given as hitting times.
	Let us first introduce the constants $z_{i,l}$, $z_{i, h}$ and $z_i$ for $i=1,2$ by
	\begin{equation} \label{eq:def_zlh}
		z_{i,l}: = \Big(\frac{m-1}{\gamma - m} \frac{4 \pi_i \lambda_i b}{a \ell (1- (1- \lambda_i)^{\gamma})} \Big)^{\frac{1}{\gamma-1}}  
		~~\mbox{and}~
		z_{i,h}: = \Big(\frac{m-1}{\gamma - m} \frac{4 \pi_i \lambda_i b}{a \ell \lambda_i^{\gamma}} \Big)^{\frac{1}{\gamma-1}},
	\end{equation}
	and
	\begin{equation} \label{eq:def_z}
		z_i
		~:=~
		\Big( \frac{4 \pi_i \lambda_i b(z_{i,l}^{1-m} - z_{i,h}^{1-m}) + a \ell (1- (1- \lambda_i)^{\gamma}) z_{i,l}^{\gamma -m} - a \ell \lambda_i^{\gamma} z_{i,h}^{\gamma -m} }{a \ell (1-\lambda_1^{\gamma}- \lambda_2^{\gamma})}  \Big)^{\frac{1}{\gamma-m}},
	\end{equation}
	Notice that $\lambda_1^{\gamma} + \lambda_2^{\gamma} <1$ so that it holds that $z_{i,l} < z_i < z_{i, h}$ for each $i=1, 2$.
	Further, let us define the stopping times
	$$
		\tau_{i,l} := \inf \big \{ t \ge 0 ~: X_t \ge z_{i,l} \big\},
		~~
		\tau_{i,h} := \inf \big \{ t \ge 0 ~: X_t \ge z_{i,h} \big\},
		~~ i = 1,2.
	$$

	\begin{Theorem} \label{thm:NashE}
		Assume that $x_0 < \min(z_{1,l}, z_{2,l})$.

		\vspace{0.5em}

	\noindent $\mathrm{(i)}$ 
		When $z_1 = z_2$, 
			there are two equilibrium solutions: 
			$$
				( \tau_{1,h}, \tau_{2,l} )
				~\mbox{satisfying}~\tau_{2,l} < \tau_{1,h};
				~~\mbox{and}~~
				( \tau_{1,l} , \tau_{2,h})
				~\mbox{satisfying}~\tau_{1,l} < \tau_{2,h}.
			$$
	\noindent $\mathrm{(ii)}$ When $z_1 < z_2$, 
		\begin{itemize}
			\item[-] if $z_{2,l} \leq z_1 < z_2 \leq z_{1,h}$, there are two equilibrium solutions:
			$$
			( \tau_{1,h}, \tau_{2,l} )
			~\mbox{satisfying}~\tau_{2,l} < \tau_{1,h};
			~~\mbox{and}~~
			( \tau_{1,l} , \tau_{2,h})
			~\mbox{satisfying}~\tau_{1,l} < \tau_{2,h}.
			$$
			\item[-]   otherwise, there is only one equilibrium solution $ ( \tau_{1,l} , \tau_{2,h}) ~\mbox{satisfying}~\tau_{1,l} < \tau_{2,h}$. 
		\end{itemize}
	
	\noindent $\mathrm{(iii)}$  When $z_2 < z_1$, 	
	\begin{itemize}
		\item[-] if $z_{1,l} \leq z_2 < z_1 \leq z_{2,h}$, there are two equilibrium solutions:
		$$
		( \tau_{1,h}, \tau_{2,l} )
		~\mbox{satisfying}~\tau_{2,l} < \tau_{1,h};
		~~\mbox{and}~~
		( \tau_{1,l} , \tau_{2,h})
		~\mbox{satisfying}~\tau_{1,l} < \tau_{2,h}.
		$$
		\item[-] otherwise, there is only one equilibrium solution $( \tau_{1,h}, \tau_{2,l}) ~\mbox{satisfying}~\tau_{2,l} < \tau_{1,h}$. 
	\end{itemize}

	\noindent $\mathrm{(iv)}$ Finally, for the Nash equilibrium $ ( \tau_{1,l} , \tau_{2,h})$, one obtains the total social utility:
	\begin{align*}
		u(x_0) = u_1(x_0)
			~:=~& 
			\big (-2\pi_1 \lambda_1 b z_{1,l} - a \ell (1-\lambda_2^{\gamma})  z_{1,l}^{\gamma} \big) \big( \frac{x_0}{z_{1,l}} \big)^m 
			+
			\big(-2\pi_2 \lambda_2 b z_{2,h} - a \ell \lambda_2^{\gamma}  z_{2,h}^{\gamma} \big) \big(\frac{x_0}{z_{2,h}} \big)^m \\
			&+ a \ell x_0^{\gamma} + (\pi_1 \lambda_1 +\pi_2 \lambda_2) b x_0;
	\end{align*}
	and for the Nash equilibrium $( \tau_{2,l}, \tau_{1,h})$, the total social utility is
	\begin{align*}
			u(x_0) = u_2(x_0) 
			~:=~&
			\big( - 2\pi_2 \lambda_2 b z_{2,l} - a \ell (1-\lambda_1^{\gamma}) z_{2,l}^{\gamma} \big) \big( \frac{x_0}{z_{2,l}}  \big)^m 
			+
			\big( -2 \pi_1 \lambda_1 b {z}_{1,h}- a \ell \lambda_1^{\gamma} {z}_{1,h}^{\gamma} \big) 
			\big( \frac{x_0}{{z}_{1,h}} \big)^m \\
			&+ a \ell x_0^{\gamma} + (\pi_1 \lambda_1 +\pi_2 \lambda_2) b x_0 .
	\end{align*}
\end{Theorem}

	\begin{Remark} \label{rem:spec}{\rm 
	Under the special condition that  $\gamma \ge 2$ and
	$$
		\pi_i = \alpha \lambda_i, ~i=1,2, ~\alpha > 0 
		~~\mbox{so that}~
		\lambda_1 < \lambda_2 \text{ and } 0 < \lambda_1 < \frac{1}{2},
	$$
	it always holds that
	$$
		z_{1,l} < z_{2,l} \text{ and } z_{2,h} < z_{1,h}.
	$$
	Assume in addition that $\gamma = 2$, then one has $z_1 < z_2 $, and hence there are only two cases:
	\begin{itemize}
		\item If $z_{1, l} < z_1 < z_{2, l} < z_2 < z_{2,h} < z_{1,h}$, there is only one equilibrium solution $ ( \tau_{1,l} , \tau_{2,h})$. 
		
		\item If $z_{1,l} < z_{2,l} \le z_1 < z_2 < z_{2,h} < z_{1,h}$, there are two equilibrium solutions $(\tau_{1,h}, \tau_{2,l})$ and $ ( \tau_{1,l} , \tau_{2,h})$.
	\end{itemize}
}
	\end{Remark}

\subsection{The setting with one regulator}

	We next study the problem with only one regulator for the two firms.
	The regulator can either apply individual compensation contract to incentivise the two firms to quit, or apply uniform contract to incentivise the two firms as presented in Section \ref{subsec:model}, 
	or decides the exit time of each firm without any compensation.

\subsubsection{The case with individual compensation scheme}

	As discussed in Section \ref{subsec:Nash}, in the case of individual contract, the regulator needs to pay $Z^i_{\tau_i} = \pi_i \lambda_i b X_{\tau_i}$ to incentivise Firm $i$ to exit at time $\tau_i$.
	Due to technical reasons, we will only consider the hitting times in $\Tc$, i.e.
	$$
		\Tc_0 ~:=~ \big\{ \tau \in \Tc ~: \tau = \inf \big\{ t \ge 0 ~: X_t \ge x \} ~\mbox{for some}~x> 0 \big\}
		\big\}.
	$$
	Then the regulator's problem becomes
	\begin{align}  \label{eq:individual_compen}
		\hat u(x_0) 
		~:=
		\sup_{\tau_1, \tau_2 \in \Tc_0}  
		\E \Big[ \sum_{i=1}^2 \int_0^{\tau_i}  e^{- \rho t}  \pi_i \lambda_i X_t dt - e^{-\rho \tau_i} Z^i_{\tau_i} \Big] 
		-
		D(x_0, \tau_1, \tau_2),
	\end{align}
	where $D(x_0, \tau_1, \tau_2)$ is the total social damage defined in \eqref{eq:def_Damage2}.

	\vspace{0.5em}
	
	Recall that $z_i$, $z_{i,l}$, $z_{i,h}$ for $i=1,2$ are defined in \eqref{eq:def_zlh} and \eqref{eq:def_z},
	let 
	$$
		\big( \hat z_i, ~\hat z_{i,l}, ~\hat z_{i,h} \big) 
		~:=~
		2^{ - \frac{1}{\gamma - 1}} ~\big(z_i, ~z_{i,l}, ~z_{i,h} \big), 
		~~i=1,2,
	$$
	and 
	$$
		\hat \tau_{i,l} := \inf \big \{ t \ge 0 ~: X_t \ge \hat z_{i,l} \big\},
		~~
		\hat \tau_{i,h} := \inf \big \{ t \ge 0 ~: X_t \ge \hat z_{i,h} \big\},
		~~ i = 1,2.
	$$
	
	 \begin{Proposition} \label{prop:indiv_contract}
	 	Assume that $x_0 < \min (\hat z_{1,l}, \hat z_{2,l})$.
	 
	 	\vspace{0.5em}
		\noindent $\mathrm{(i)}$ When ${\hat z_1} < {\hat z_2}$, the optimal solution of the problem \eqref{eq:individual_compen} is given by $({\hat \tau_{1,l}}, {\hat \tau_{2,h}})$,
		and the total social utility is
			\begin{align*}
				\hat u(x_0) = \hat u_1(x_0) 
				:=&
				\big (-2\pi_1 \lambda_1 b \hat z_{1,l} - a \ell (1-\lambda_2^{\gamma}) \hat z_{1,l}^{\gamma} \big) \big( \frac{x_0}{\hat z_{1,l}} \big)^m 
				+
				\big(-2 \pi_2 \lambda_2 b \hat z_{2,h} - a \ell \lambda_2^{\gamma} \hat z_{2,h}^{\gamma} \big) \big(\frac{x_0}{\hat z_{2,h}} \big)^m \\
				&+ a \ell x_0^{\gamma} + (\pi_1 \lambda_1 +\pi_2 \lambda_2) b x_0.
			\end{align*}

		\noindent $\mathrm{(ii)}$ When ${ \hat z_2} < { \hat z_1}$, the optimal solution of the problem \eqref{eq:individual_compen} is given by $({ \htau_{1,h}}, { \htau_{2,l}})$,
		and the total social utility is
			\begin{align*}
				\hat u(x_0) = \hat u_2(x_0) 
				:=&
				\big( - 2\pi_2 \lambda_2 b \hat{z}_{2,l} - a \ell (1-\lambda_1^{\gamma}) \hat{z}_{2,l}^{\gamma} \big) \big( \frac{x_0}{\hat{z}_{2,l}} \big)^m 
				+
				\big( -2 \pi_1 \lambda_1 b \hat{z}_{1,h}- a \ell \lambda_1^{\gamma}\hat{z}_{1,h}^{\gamma} \big) 
				\big( \frac{x_0}{\hat{z}_{1,h}} \big)^m \\
				&+ a \ell x_0^{\gamma} + (\pi_1 \lambda_1 +\pi_2 \lambda_2) b x_0 .
			\end{align*}

	 \end{Proposition}

\subsubsection{The case of uniform compensation scheme}

	As comparison, let us consider the case of uniform compensation scheme as studied in Section \ref{sec:one-country}  (where $\pi_1 < \pi_2$).
	Applying the results in Section \ref{subsec:1market}, in case the regulator applies  the same compensation scheme to both firms $i=1, 2$, the optimal exit time of the each firm are given by
	$\htau_1$ and $\htau_2$ which are hitting times of the threshold $\hat x_1$ and $\hat x_2$, where 
	$$
		\hat x_1 := \Big( \frac{m-1}{\gamma-m} \frac{b}{\ell a} \frac{2\pi_1 \lambda_1 } {1- {\lambda}_{2}^{\gamma} } \Big)^{1/({\gamma}-1)},
		~~~
		\hat x_2 := \Big( \frac{m-1}{\gamma-m} \frac{b}{\ell a} \frac{\pi_2(1+\lambda_2) - \pi_{1} {\lambda}_{1} } {{\lambda}_{2}^{\gamma} } \Big)^{1/({\gamma}-1)}
	$$
	Moreover, the total utility is given by
	\begin{align*}
	v(x_0) = &  \big(  -2 \pi_1  \lambda_1  b \hat x_1  - a \ell (1 - \lambda_{2}^\gamma \big) \hat x_1^\gamma \big) \big( \frac{x_0}{\hat x_1}\big)^m + \big(  -2 \lambda_2 \pi_2  b \hat x_2  - \lambda_1 (\pi_2 - \pi_1) b \hat x_2 -  a \ell \lambda_{2}^\gamma  \hat x_2^\gamma \big) \big( \frac{x_0}{\hat x_2}\big)^m  \\ 
    & + a \ell x_0^{\gamma} + (\pi_1 \lambda_1 +\pi_2 \lambda_2) b x_0  ,
	\end{align*}

	\begin{Remark}{\rm 
	\noindent $\mathrm{(i)}$ When \( \hat z_1 < \hat z_2 \) (or equivalently $z_1 < z_2$), in the case of individual compensation scheme, there is only one optimal solution: Firm 1 exits first at time \( \hat \tau_{1,l} \) and Firm 2 exits next at \( \hat \tau_{2,h} \). 
	At the same time, we have
$$
\hat z_{1,l} = \hat x_1 < \hat z_{2,h} < \hat x_2 \quad \text{and} \quad \hat u_1(x_0) > v(x_0),
$$
which indicates that Firm 1 exits the market at the same time under both the individual and uniform compensation schemes, while Firm 2 exits later in the uniform compensation case. Therefore, the uniform compensation scheme requires more time for both firms to exit the market and results in lower profits, making it less efficient than the individual compensation case.

	\vspace{0.5em}

\noindent $\mathrm{(ii)}$ When \( \hat z_1 > \hat z_2 \) (or equivalently $z_1 > z_2$), in the case of individual compensation scheme, there is unique optimal solution is: 
	Firm 2 exits first at \( \hat \tau_{2,l} \) and Firm 1 exits next  at \( \hat \tau_{1,h} \). 
	One has
$$
\hat \tau_{2,l} < \hat \tau_{1,h}, \quad \hat z_{2,l} < \hat x_2, \quad \hat x_1 < \hat z_{1,h}, \quad \text{and} \quad \hat u_2(x_0) > v(x_0).
$$
Compared to the individual compensation case, Firm 1 exits the market earlier under the uniform compensation scheme, while Firm 2 exits later. 

	\vspace{0.5em}

Additionally, if \( \lambda_1 > 0.5 \), then \( \hat z_{1,h} < \hat x_2 \).
	Otherwise, it is possible that \( \hat z_{1,h} \geq \hat x_2 \) for some  parameter \( \pi_1, \, \pi_2 \).
	For example, with \( \lambda_1 = 0.1 \), \( \gamma = 2 \), \( \pi_1 = 1 \) and \( \pi_2 = 5 \), one has 
	$$\hat z_{1,h} ~=~ 20 \frac{m-1}{2-m} \frac{b}{\ell a}  ~>~ \hat x_2 ~=~ 12   \frac{m-1}{2-m} \frac{b}{\ell a}.$$
	In this case, the time required for all firms to exit in the individual compensation case is not necessarily less than in the uniform compensation case.
	However, the total utility of individual compensation scheme remains higher than in the uniform compensation case.
	}
	\end{Remark}

\subsubsection{The social optimum without compensation}

	We next consider the case where the regulator can decide the exit time of each firm without given any compensation.
	Again, we will only consider hitting times in this setting, so that the regulator's optimal stopping problem becomes
	\begin{align} \label{eq:central_2}
		\bar u(x_0) 
		~:=
		\sup_{\tau_1, \tau_2 \in \Tc_0}
		\E \Big[ \int_0^{\tau_1}  e^{- \rho t}  \pi_1 \lambda_1 X_t dt + \int_0^{\tau_2}  e^{- \rho t}  \pi_2 \lambda_2 X_t dt \Big] - D(x_0, \tau_1, \tau_2),
	\end{align}
	where the  $D_0(x_0, \tau_1, \tau_2)$ is the total social damage defined in \eqref{eq:def_Damage2}.
	
	\vspace{0.5em}
	
	Let us define
	$$
		\big(\bar  z_i,~ \bar z_{i,l}, ~\bar z_{i,h} \big) 
		~:=~
		4^{ - \frac{1}{\gamma - 1}} ~\big(  z_i, ~z_{i,l}, ~z_{i,h} \big),
		~~i=1,2,
	$$
	where $z_i$, $z_{i,l}$, $z_{i,h}$ for $i=1,2$ are defined in \eqref{eq:def_x1_p} and \eqref{eq:def_x_1},
	and then define
	$$
		\bar \tau_{i,l} := \inf \big \{ t \ge 0 ~: X_t \ge \bar z_{i,l} \big\},
		~~
		\bar \tau_{i,h} := \inf \big \{ t \ge 0 ~: X_t \ge \bar z_{i,h} \big\},
		~~ i = 1,2.
	$$

	 \begin{Proposition}
	 	Assume that $x_0 < \min( \bar z_{1,l}, \bar z_{2,l})$.
		
		\vspace{0.5em}
	 
	 	\noindent $\mathrm{(i)}$ 
		When ${\bar z_1} < { \bar z_2}$, the optimal solution of the problem \eqref{eq:central_2} is given by $({\bar \tau_{1,l}}, {\bar \tau_{2,h}})$,
		and the total social utility is
			\begin{align*}
				\bar u(x_0) = \bar u_1(x_0) 
				~:=~&
				\big (-\pi_1 \lambda_1 b \bar z_{1,l} - a \ell (1-\lambda_2^{\gamma}) \bar z_{1,l}^{\gamma} \big) \big( \frac{x_0}{\bar z_{1,l}} \big)^m 
				+
				\big(-\pi_2 \lambda_2 b \bar z_{2,h} - a \ell \lambda_2^{\gamma} \bar z_{2,h}^{\gamma} \big) \big(\frac{x_0}{\bar z_{2,h}} \big)^m  \\
				&+ a \ell x_0^{\gamma} + (\pi_1 \lambda_1 +\pi_2 \lambda_2) b x_0.
			\end{align*}

		\noindent $\mathrm{(ii)}$
		When ${\bar z_2} < {\bar  z_1}$, the optimal solution of the problem \eqref{eq:central_2} is given by $({\bar \tau_{1,h}}, {\bar  \tau_{2,l}})$,
		and the total social utility is
			\begin{align*}
				\bar u(x_0) = \bar u_2(x_0) 
				~:=~&
				\big( - \pi_2 \lambda_2 b \bar{z}_{2,l} - a \ell (1-\lambda_1^{\gamma}) \bar{z}_{2,l}^{\gamma} \big) \big( \frac{x_0}{\bar{z}_{2,l}} \big)^m 
				+
				\big(-\pi_1 \lambda_1 b \bar{z}_{1,h}- a \ell \lambda_1^{\gamma}\bar{z}_{1,h}^{\gamma} \big) 
				\big( \frac{x_0}{\bar{z}_{1,h}} \big)^m \\
				&+ a \ell x_0^{\gamma} + (\pi_1 \lambda_1 +\pi_2 \lambda_2) b x_0 .
			\end{align*}
	 \end{Proposition}

	\begin{Remark}{\rm 
		For one regulator scenario with individual compensation and without compensation cases, we have
		$$
		\big( \bar z_{i,l}, ~\bar z_{i,h} \big) 
		~:=~
		2^{ - \frac{1}{\gamma - 1}} ~\big( \hat z_{i,l}, ~ \hat z_{i,h} \big), ~~ 
		\hat u_i(x_0) < \bar u_i(x_0).
		~~i=1,2,
		$$
		This is reasonable because the regulator requires more time to balance his/her interests and may need to sacrifice some interests in the compensation scenario.
		}
	\end{Remark}

	Next, we compare the scenario with two regulators to the cases with individual compensation and without compensation.
	\begin{Remark}{\rm 
	Consider the special condition where  $\gamma = 2$ and
		$$
		\pi_i = \alpha \lambda_i, ~i=1,2, ~\alpha > 0 
		~~\mbox{so that}~
		\lambda_1 < \lambda_2 \text{ and } 0 < \lambda_1 < \frac{1}{2}.
		$$
	\begin{itemize}
	\item[$\mathrm{(i)}$] When \( z_1 < z_2 \), the scenario with two regulators achieves a Nash equilibrium at Firm 1 with \( \tau_{1,h} \) and Firm 2 with \( \tau_{2,l} \). In the cases of individual compensation and without compensation, the optimal solutions are attained at \( (\hat{\tau}_{1,l}, \hat{\tau}_{2,h}) \) and \( (\bar{\tau}_{1,l}, \bar{\tau}_{2,h}) \), respectively. We have:
	$$
	\bar{z}_{1,l} < \hat{z}_{1,l} < z_{2,l}, \quad \bar{z}_{2,h} < \hat{z}_{2,h} < z_{1,h},
	$$
	$$
	u_2(x_0) < \hat{u}_1(x_0) < \bar{u}_1(x_0).
	$$
	\item[$\mathrm{(ii)}$] When \( z_1 < z_2 \), the case with two regulators reaches a Nash equilibrium at Firm 1 with \( \tau_{1,l} \) and Firm 2 with \( \tau_{2,h} \). The cases with individual compensation and without compensation obtains the optimal solution at \( (\hat{\tau}_{1,l}, \hat{\tau}_{2,h}) \) and \( (\bar{\tau}_{1,l}, \bar{\tau}_{2,h}) \) respectively. We find:
	$$
	\bar{z}_{1,l} < \hat{z}_{1,l}  < z_{1,l}, \quad \bar{z}_{2,h} < \hat{z}_{2,h} < z_{2,h}
	$$
	\[
	u_1(x_0) < \hat{u}_1(x_0) < \bar{u}_1(x_0).
	\]

	\item[$\mathrm{(iii)}$]When $z_1 > z_2$, the two regulators scenario achieves a Nash equilibrium at Firm 1 with $\tau_{1,l}$ and Firm 2 with $\tau_{2,h}$. The cases with individual compensation and without compensation achieve optimal solutions at $(\hat \tau_{1,h}, \hat \tau_{2,l})$ and $(\bar \tau_{1,h}, \bar \tau_{2,l})$, respectively. We have:
	$$
	 \bar z_{2,l} < \hat z_{2,l} <z_{1,l} , \quad \bar z_{1,h} <\hat z_{1,h} < z_{2,h}
	$$
	$$
	u_1(x_0) <\hat u_2(x_0) < \bar u_2(x_0).
	$$
	\item[$\mathrm{(iv)}$]When $z_1 > z_2$,  the two regulators case achieves a Nash equilibrium at Firm 1 with $\tau_{1,h}$ and Firm 2 with $\tau_{2,l}$. The cases with individual compensation and without compensation obtain optimal solution at $(\hat \tau_{1,h}, \hat \tau_{2,l})$ and $(\bar \tau_{1,h}, \bar \tau_{2,l})$, respectively. We find:
	$$
	 \bar z_{2,l} < \hat z_{2,l}  < z_{2,l}  ,  \quad \bar z_{1,h} <\hat z_{1,h} <z_{1,h}
	 $$
	 $$
	u_2(x_0) <\hat u_2(x_0) < \bar u_2(x_0).
	$$
\end{itemize}
	From the above comparison, we can see that the timing of the first exit for firms follows this order: one regulator without compensation, one regulator with individual compensation, and then the case with two regulators. The timing of the second exit for firms exhibits the same order. However, the utility is in the reverse order. This indicates that the scenario with one regulator without compensation consistently leads to the earliest exit for firms and maximizes utility, making it the most efficient case. Conversely, the scenario with two regulators results in the latest exits and the lowest utility, likely due to the competition between the two regulators causing additional losses.
	}
	\end{Remark}

	\begin{Remark}{\rm 
	Theoretically, we can extend our analysis to multiple (more than two) interacting markets for all models in this section and apply the same approach. However, since the two interacting market cases effectively illustrate the main ideas and provide clear results, further discussion would involve tedious repetition. Therefore, we will not address that.
	}
	\end{Remark}

\section{Economic application}
\label{sec:eco_application}

To illustrate the consequences of the above incentive mechanism to exit the market, we consider the crude oil industry. Its production is roughly 100~Mb/d (million barils per day). If no change is done to curb its production, we can assume an average growth rate of $\mu =2\%$ per year and a volatility of $\sigma = 0.08$.  We assume an average price of $50$ USD/baril. We assume a damage paramater of $\gamma=4$, which is large. And we set $\ell$ so that the damages are already outweighting the profit, i.e set $\ell$ such that $p X_0 - \ell X_0^{\gamma} = 0$ where $X_0 =$ 100~Mb/d. We will assume that the market shares and the profits are given by the relation $\pi_i = \alpha \lambda_i$ and that $\lambda_i = K_n/(N+1-i)^\theta$. We first consider the case of a single market served by $N$ firms and then, consider the interaction between two countries managing their own market.


\subsection{Single market}
 
We illustrate the effect of the discount rate, of the number of firms serving the market and of the market concentration on the time and costs required to halt crude production. 

Regarding the effect of the discount rate, we consider that the whole crude market is served by a single company. In that case, we can compare easily the socialy optimal threshold of production to stop crude production $\bar x_1$ and the second-best incentive payment threshold $\hat x_1 = 2^{\frac{1}{\gamma-1}} \bar x_1$ as well as the induced costs. In our model, damages and production profit are discounted at the same rate $\rho$. While the discount rate in the oil industry tends to be quite high where rates of $15$\% are not unusual due to the risk involved in the exploration/expoitation of oil wells, the Stern's Review on Climate Economics (\cite{Stern06}) used a much lower rate of $1.4$\% based on endogeneous long term growth model. The Table below resumes the relevent economic variable for $\rho=3$\% and for $\rho=10$\%. We observe that a large discount rate delays the moment to stop the production both in the first and the second best situation because large discounting factor reduces the value of the future damages. But, at the same, it also reduces the value of future profits from production and thus, it considerably reduces the amount to be paid to the firm in monopoly.

\begin{table}[h]
\centering
\begin{tabular}{c c c c c c c c } 
   & $\bar x_1$ & $\hat x_1$ & $u(X_0)$ & $v(X_0)$ & $\E[e^{-\rho \hat \tau_1} \hat Y_{\hat \tau_1}]$ & $\E[\bar \tau_1]$ & $\E[\hat \tau_1]$      \\ 
   & Mb/d  & Mb/d   & trillions  & trillions  & trillions & y  &  y  \\ \hline
&   &    &  &  &  &   &    \\  
$\rho=3$\%    &  112  & 141   & 1.2  &  $-$164   & 159   & 6.5  & 20  \\ 
&   &    &  &  &  &   &    \\  
$\rho=10$\%  & 109  & 138  &   0.6 &  $-$13   &  10  &  5 &  19  \\ 
\end{tabular}
\caption{Economic indicators in the case of a single firm serving the market. Parameters: $\mu=0.02$, $\sigma=0.08$, $\pi=50$~M\$/Mbaril, $\gamma=4$, $\ell=1.02\,10^{-6}$.}
\label{tab:params}
\end{table}

\begin{figure}[htb!]
\center
\begin{tabular}{c c}
\includegraphics[width=0.45\textwidth]{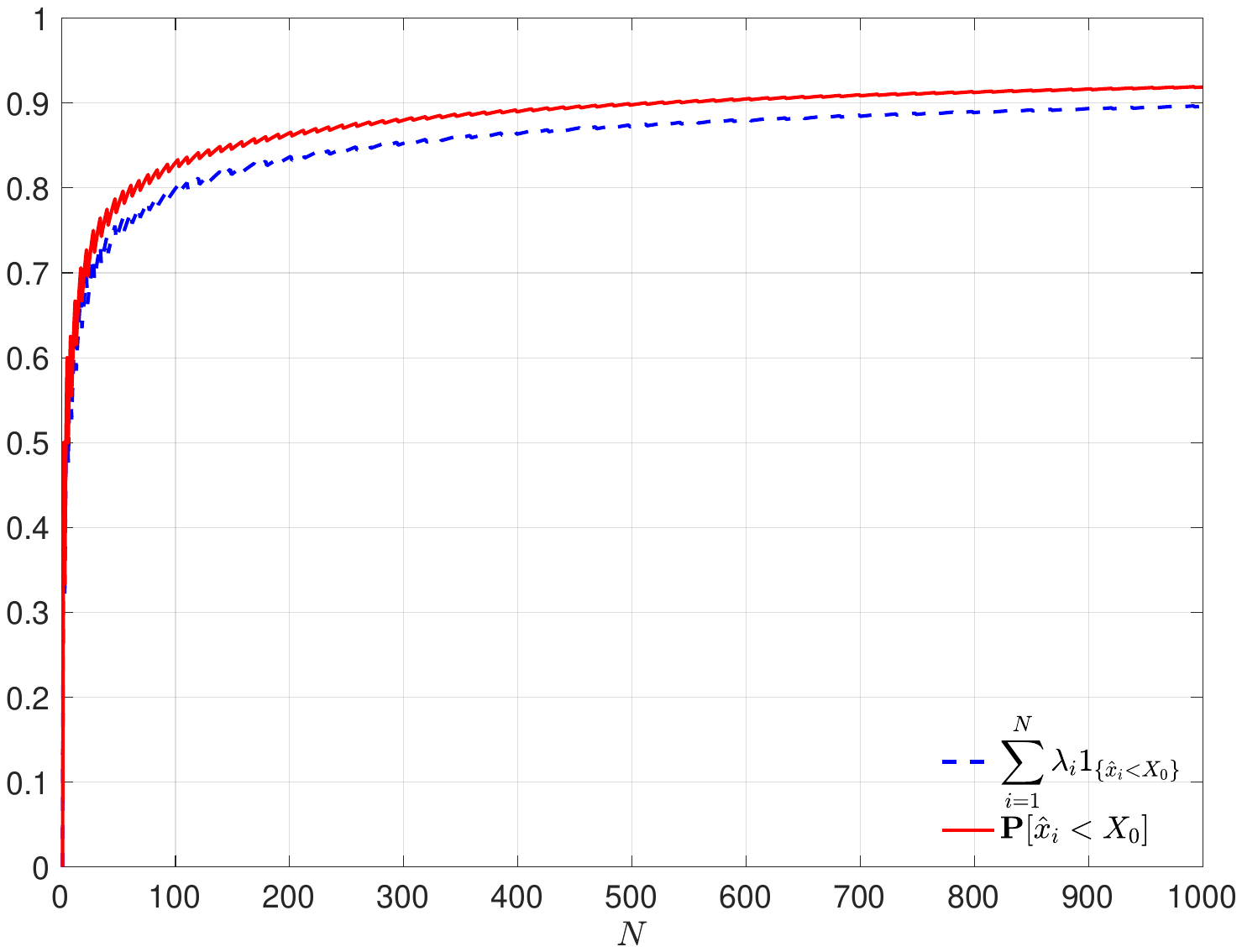} &
\includegraphics[width=0.45\textwidth]{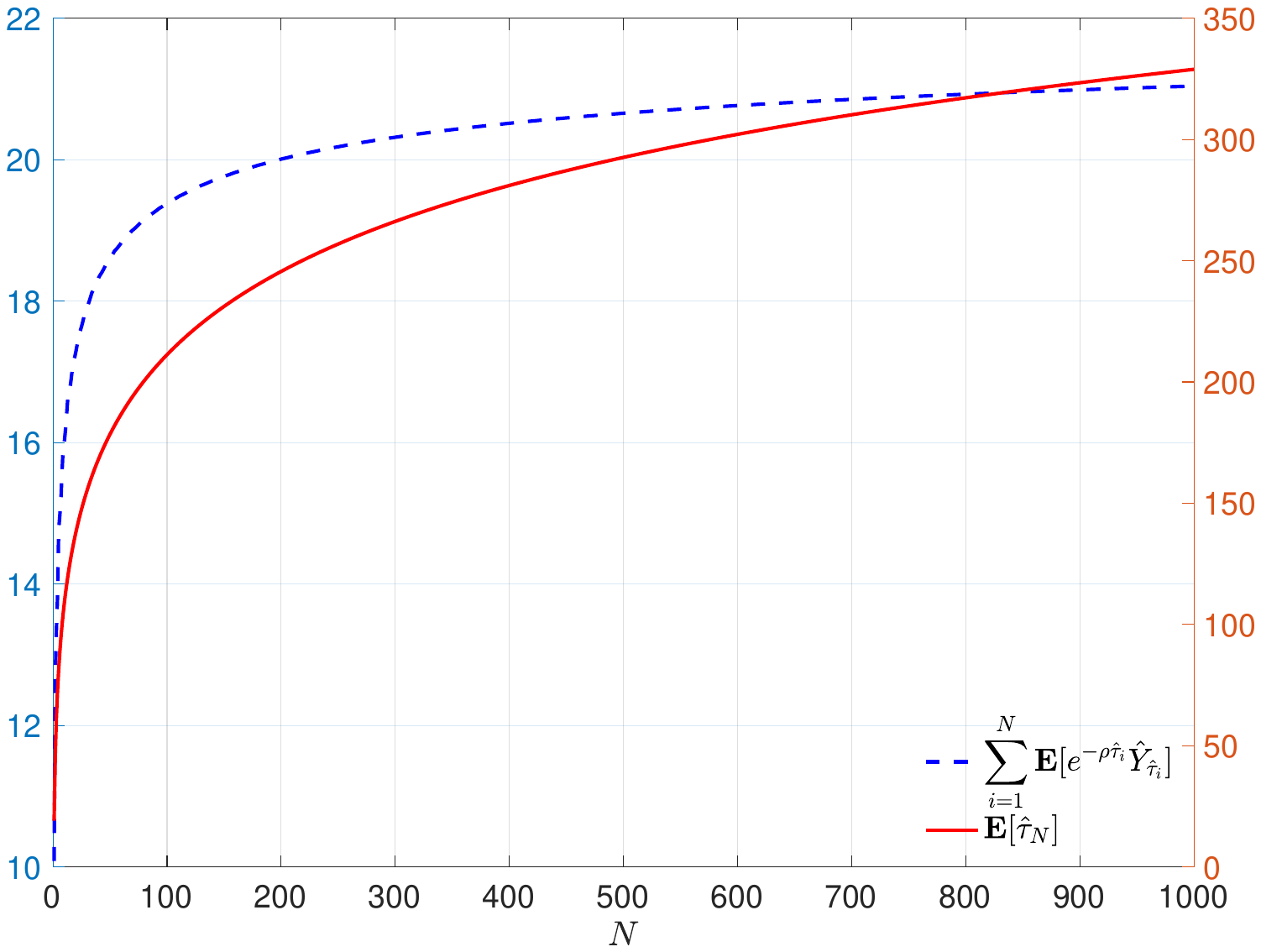}
\end{tabular}
\caption{{\small As a function of the number of firms $N$ (Left) Fraction of firms and of market share leaving immediatly the market. (Right) Total expected discounted payment (in trillions of dollars, left axis) and expect time of closing all production (in years, right axis). Parameters: $\mu=0.02$, $\sigma=0.08$, $\rho=0.1$, $\gamma=4$, $\pi =50\,10^6$, $\ell = 1.02\, 10^{-6}$, $\theta=0.1$. } }
\label{fig:LbdY-1}
\end{figure}

\begin{figure}[htb!]
\center
\begin{tabular}{c c}
\includegraphics[width=0.45\textwidth]{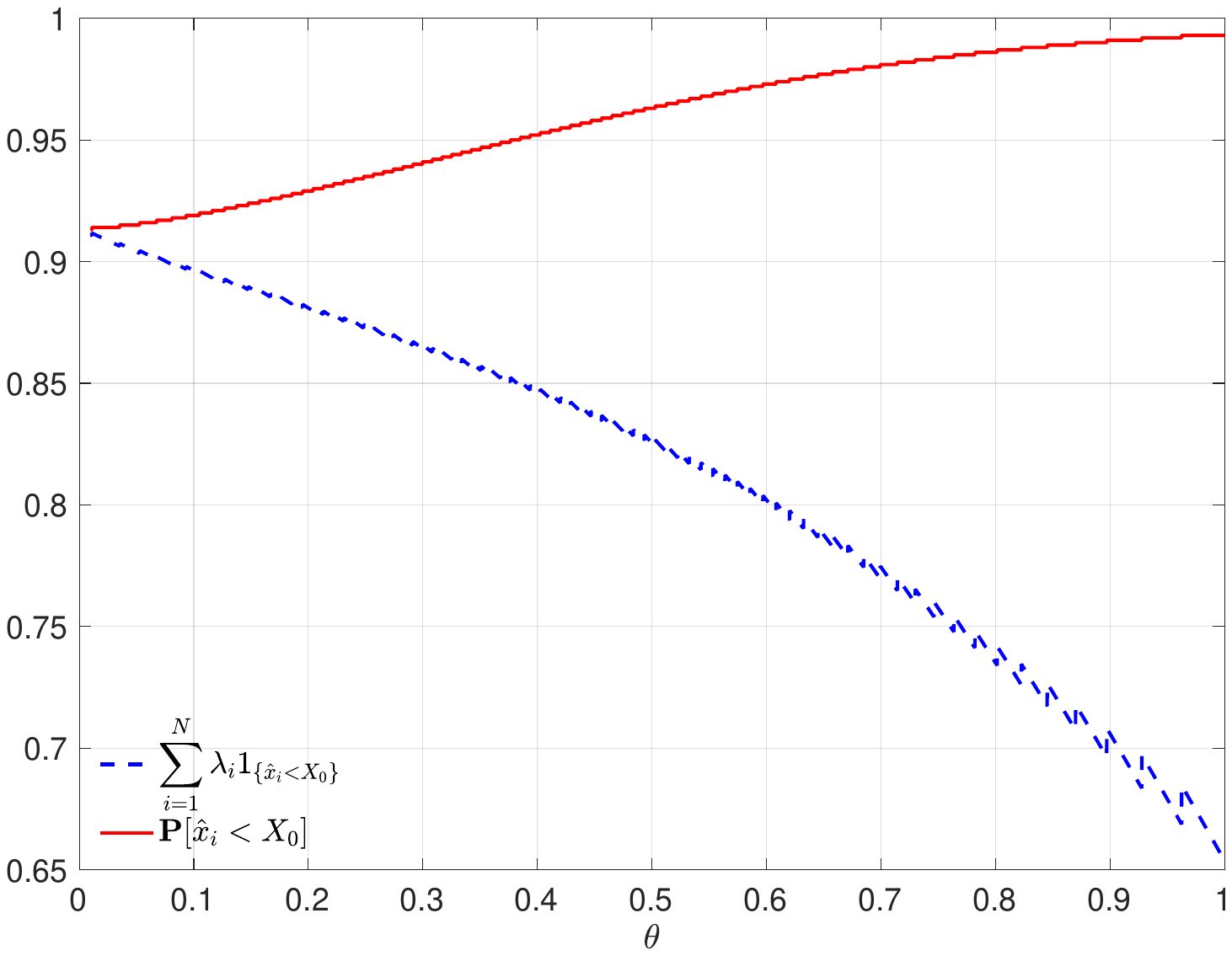}  &
\includegraphics[width=0.45\textwidth]{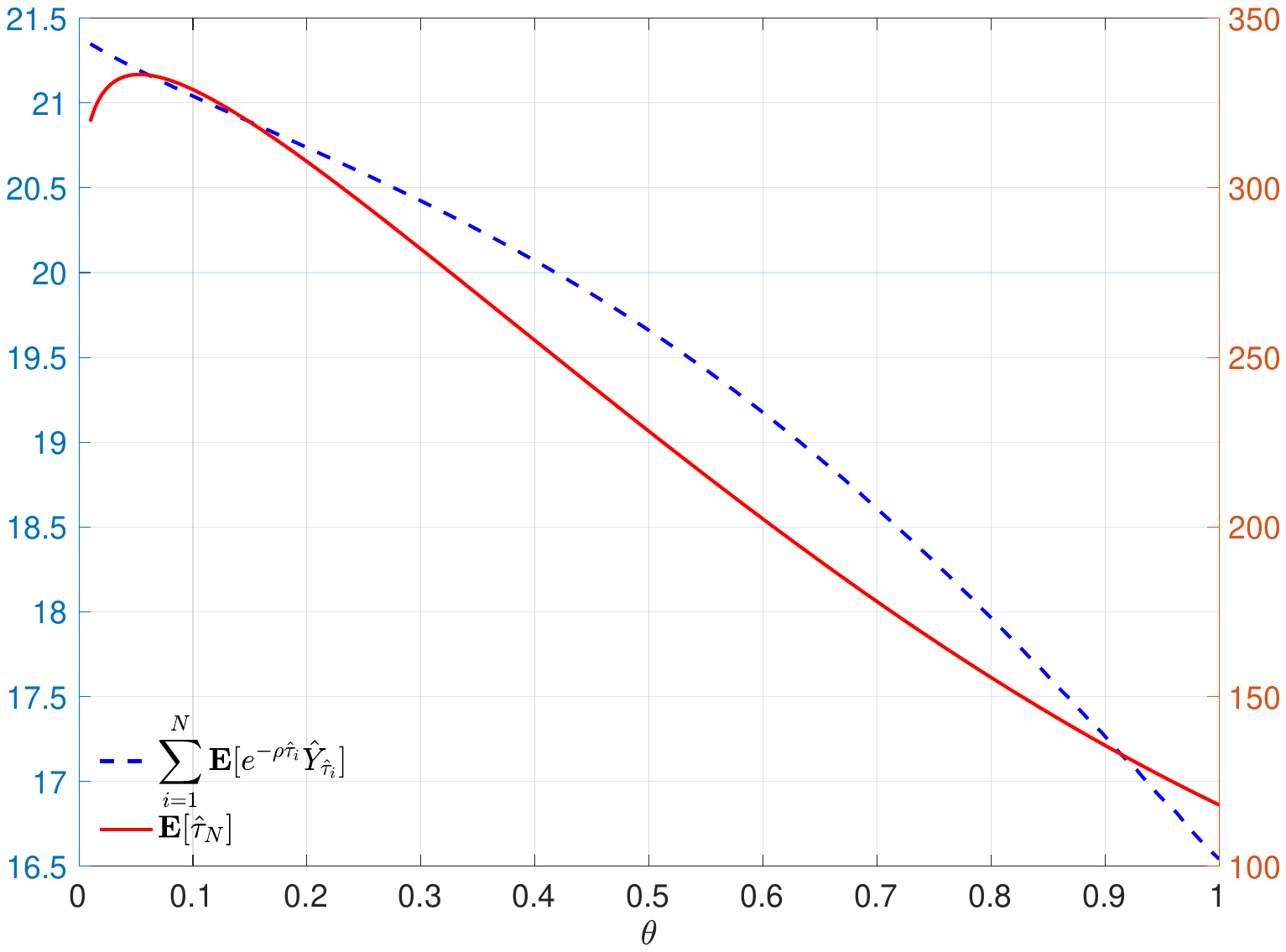} 
\end{tabular}
\caption{{\small As a function of the tail parameter $\theta$ of the size distribution of firms (Left) Fraction of firms and of market share leaving immediatly the market. (Right) Total expected discounted payment (in trillions of dollars, left axis) and expect time of closing all production (in years, right axis). Parameters: $\mu=0.02$, $\sigma=0.08$, $\rho=0.1$, $\gamma=4$, $\pi =50\,10^6$, $\ell = 1.02\, 10^{-6}$, $N=10^3$. } }
\label{fig:LbdY-2}
\end{figure}

Figure~\ref{fig:LbdY-1} illustrates the effect of an increasing number of firms serving the market. To limit the effect of market concentration as the number of firms increases, we take $\theta=0.1$ so that the distribution of market is close to uniform but yet satisfies the condition $\lambda_i < \lambda_j$ for $i<j$. We observe on the left panel the fraction of firms, and the fraction of market share they represent, that are invited to immediatly leave the market. When there are $100$ firms, almost $70$ are exiting the market and they represent as expected $70$\% of the market. Total payment is an increasing function of the number of firms, yet the marginal payment is decreasing. Regarding expected time to halt all production, observe how with even $50$ firms, it requires more than two centuries. Indeed, once a large portion of the production has been stopped ($\approx$ 80\%), the damages induces by the remaining part is sufficiently reduced to allow to wait before closing it down.

Regarding the effects of market concentration, Figure~\ref{fig:LbdY-2} provides the same indicator as in Figure~\ref{fig:LbdY-1} but in this case we have fixed the number of firms to $N=10^3$ and made $\theta$ varies from $0.1$ (no market concentration) to $1$ highly concentrated market. We observe that when $\theta$ gets closer to $1$,  almost all firms are required to leave the market, but this large number of firms represent an ever decreasing amount of production. With $\theta=1$,  $65$\% of total production is eliminated by exiting $99$\% of firms. In terms of payments and expected time to close production, the total expected discounted payment is 21 trillions dollars for one thousand firms and $\theta=0.1$ whereas it is reduced to 16.5 trillions for the same number of firms with a market concentration parameter $\theta=1$. We observe the same reduction in the expected time which is divided by $3$. 

As a result, market concentration makes easier the work of the regulator who wishes to stop a polluting activity.

\subsection{Two markets}

Even if the results on the interaction between regulators in Section~\ref{sec:two-countries} are limited to two markets, our model can provide insights on the interaction of large oil producing countries. Thus, let us consider the two oil producers in the world, namely the USA and Saudi Arabia with a respective production of 13~Mb/d and 10~Mb/d. Using the same parameters value as in the section above, we give in Table~\ref{tab:2countries} the exit thresholds in the second and first best.

\begin{table}[h]
\centering
\begin{tabular}{c c c | c c c || ccc | ccc } 
  \multicolumn{6}{c||}{Competition between countries} & \multicolumn{6}{c}{Social optimum without compensation} \\
   $z_{1,l}$ & $z_1$ & $z_{1,h}$ &  $z_{2,l}$ & $z_2$ & $z_{2,h}$ & $\bar z_{1,l}$ & $\bar z_1$ & $\bar z_{1,h}$   & $\bar z_{2,l}$ & $\bar z_2$ & $\bar z_{2,h}$     \\ 
  &               &           &             &          &         &        &        &         &          &        &      \\  
 135 & 175   & 292   &    119    & 160  & 339  & 85 & 175  & 184  & 75  & 160 & 213    
\end{tabular}
\caption{Exit thresholds expressed in Mb/d in the case of two interacting countries. Parameters: $\rho=0.1$, $\mu=0.02$, $\sigma=0.08$, $\pi=50\,10^6$~M\$/Mbaril, $\gamma=4$, $\ell=1.02\,10^{-6}$. $\lambda_1=13/23$, $\lambda_2=10/23$. }
\label{tab:2countries}
\end{table}

 In the social optimum without compensation situation which can be considered as the first-best social optimum, it holds that $\bar z_2 < \bar z_1$ and the optimum is thus given by $(\bar z_{1,h}, \bar z_{2,l})$. Thus, Saudi Arabia should immediatly stop since $\bar z_{2,l}$ is already lower than the current value of crude of 100~Mb/d. On the other hand, the USA could have an expected time before closing all wells of $36$ years ($\ln(184/100)/(\mu-\frac12\sigma^2)$). Regarding the competition between countries, we are in the situation (ii) of Proposition~\ref{thm:NashE} where $z_1<z_2$ and $z_{2,l} < z_1 < z_2 < z_{1,h}$. Thus, there are two equilibria: either Saudi Arabia stops first or USA stops first. This is a standard example in games of stopping of second-mover advantage. In the first equilibrium,  Saudi Arabia would run its oil production for an expected lifetime of $10$ years and the USA for $63$ years. In the second, the USA would stop production in $17$ years and Saudi Arabia in $72$ years. It is unlikely that any of these two countries would prefer the equilibrium where it stops first.  Thus, in the present modeling framework, there are no reason for the oil production to stop. Indeed, alternative incentive mechanisms could be considered that would help resolve this second-mover advantage. In particular, one could introduce compensation mechanism between countries. The 2005 Paris Agreement under Article~9 establishes a Green Climate Fund that can be considered and embryo of those types of compensation scheme. We leave for future research the introduction of cross countries optimal incentive payments to forster the removal of brown assets.

\section{Conclusion}

In this paper, we analysed how a mechanism based on incentive payment would lead to the progressive closing of polluting firms. This mechanism simply compensate firms for their lost future income. 
When there are several firms  (whose profit rates are proportional to their size) and the payment is to be the same  per unit of production for any firm, we show that small firms should exit first, leaving large polluting firms being removed last. To this regard, it can be understood that the market concentration is favorable to reducing the cost of closing down a polluting activity scattered amongst a large number of firms. Indeed, it is easier to close down a monopoly than a highly fragmented market. We also provided the equilibrium between interacting countries polluting each other and showed evidence of possible multiple equilibria sustained by a second-mover advantage. In the absence of cross payment between countries, it is thus unlikely that countries will agree on an orderly manner of closing down carbon emissive assets. 

\appendix

\section{Appendix}	
	Recall that $X$ follows the  geometric Brownian motion dynamic, with constants $\sigma > 0$ and $\mu > \frac{\sigma^2}{2}$,
	$$
	dX_t ~=~ \mu X_t dt + \sigma X_t dW_t,
	$$
	and $\mathcal{T}$ is the set of all stopping times w.r.t. the filtration $\mathbb{F}$ generated by $X$ such that
	$$
		 \E(e^{-\rho \tau}X_{\tau}^{\gamma} \log^+(e^{-\rho \tau} X_{\tau}^{\gamma})) <\infty,
	$$
	where $\log^+x = \max(\log x,0).$ 
	We consistently employ $C$ to represent positive constants whose value may differ from one to another.

\subsection{A free boundary problem and its solution}
\label{app:comp}
	
	Let $\alpha$, $\beta>0$, $\gamma \geq 2$, $M:= \gamma \mu + \frac{1}{2} \sigma^2  (\gamma^2 - \gamma)$, $\rho \in ( \mu, M)$,
	and the constants $a$, $b$, $m$ are defined as in \eqref{eq:def_abm}, i.e.
	$$
		a := \frac{1}{\gamma [\mu + \frac{1}{2} \sigma^2 (\gamma - 1)] -\rho} >0,
		~~
		b:= \frac{1}{\rho-\mu} >0,
		~~
		m ~:=~ \frac{1}{2} - \frac{\mu}{\sigma^2} + \sqrt{ \Big( \frac{1}{2} - \frac{\mu}{\sigma^2} \Big)^2 +\frac{2\rho}{\sigma^2} }.
	$$
	Further, let $a_1, \, a_2, \, a_3$ be constants satisfying $a_1\leq 0$ and $\frac{\beta}{M-\rho} - a_3 >0$.  
	We consider the following optimal stopping problem:
	\begin{align} \label{eq:OT_A}
		u(x) ~:=~\sup_{\tau \in \mathcal{T}} 
		\E\Big[  \int_{0}^{\tau} e^{-\rho t}\big[\alpha X_t - \beta X_t^{\gamma} \big] \dd t + e^{-\rho \tau} (a_1 X_{\tau}+a_2 X_{\tau}^{m} +a_3 X_{\tau}^{\gamma}) ~\Big| X_0 = x \Big],
		~~x > 0.
	\end{align}
	Define 
	$$
		x^{\ast} ~:=~ \Big( \frac{m-1}{{\gamma}-m} \frac{{\alpha} b- a_1}{  \beta a -a_3}  \Big)^{1/({\gamma}-1)},
	$$
	and
	$$
		f(x) := \alpha x - \beta x^{\gamma}, ~~ g(x) := a_1 x + a_2 x^{m} + a_3 x^{\gamma},
		~~x \ge 0.
	$$

	\begin{Proposition} \label{prop:OT_solution} 
		In the above setting, the value function of the optimal stopping problem \eqref{eq:OT_A} is given by
		$$ u(x)  =
			\Big( \alpha bx + \beta a x^{\gamma} +  ( g(x^{\ast})  - \alpha bx^{\ast} - \beta a (x^{\ast})^{\gamma} ) \Big( \frac{x}{x^{\ast}}\Big)^m \Big) \1_{\{x \in (0, x^{\ast})\}} +  g(x) \1_{\{x \in [x^{\ast}, \infty)\}}.
		$$
		Moreover, the hitting time $\htau := \inf\{ t \ge 0 ~: X_t \ge x^{\ast} \}$ is an optimal solution to the optimal stopping problem \eqref{eq:OT_A}.
	\end{Proposition}

	\begin{proof} 
		Notice that $-\rho+\mu m +\frac{\sigma^2}{2}(m^2-m) = 0$. 
		It follows by Itô's formula that, for any $\tau \in \mathcal{T}$,
		\begin{align*}
			&\E\Big[  \int_{0}^{\tau} e^{-\rho t}\big[\alpha X_t - \beta X_t^{\gamma} \big] \dd t + e^{-\rho \tau} (a_1 X_{\tau}+a_2 X_{\tau}^{m} +a_3 X_{\tau}^{\gamma}) \Big] \\ 
			= & \E\Big[  e^{-\rho \tau} ((a_1-b\alpha) X_{\tau}+(-\beta a +a_3) X_{\tau}^{\gamma}) \Big]+\beta a x^{\gamma}+a_2 x^m+ b \alpha x.
		\end{align*}
		Therefore, the value function in \eqref{eq:OT_A} is equivalent to
		\begin{equation}  \label{eq:OT_B}
			v(x)~:= ~ \sup_{\tau \in \mathcal{T}} 	\E\Big[  e^{-\rho \tau} ((a_1-b\alpha) X_{\tau} +(-\beta a +a_3) X_{\tau}^{\gamma}) \Big].
		\end{equation}
		By classical result (see e.g. Theorem 5.2.1 of \cite{Pham}),
		the value function $v(x)$ in \eqref{eq:OT_B} is a viscosity solution to variational inequality
		\begin{equation} \label{ineq:var}
			\text{min} [ \rho h - \mathcal{L} h , h - g_v]=0 \text{ on }\mathbb{R}^+, 
		\end{equation}
		where $ \mathcal{L} h = \mu x  h'(x)+ \frac{1}{2} \sigma^2 x^2 h''(x)$ and $g_v(x) = (a_1-b\alpha) x +(-\beta a+a_3)x^{\gamma}.$
		
		Let us define the stopping region $\mathcal{S}:= \{ x \in \mathbb{R}^+: v(x) = g_v(x) \} $ and the continuation region $\mathcal{D}:= \{ x  \in \mathbb{R}^+ ; v(x) > g(x) \}$. Notice that $x \in \mathcal{S}$ implies $x $ is a local minimum such that $\min(v-g_v) = (v-g_v)(x)=0$ due to $v(x)\geq g_v(x)$. Since $g_v$ is $C^2$ in $x \in (0,\infty)$ and $v$ is a viscosity supersolution of \eqref{ineq:var}, $\rho g_v(x) - \mathcal{L} g_v(x)  \geq 0$ must hold. By its definition, it is easy to see that $\rho g_v(x) - \mathcal{L} g_v(x) \geq 0$ if and only if $x \geq \Big( \frac{a}{b} \frac{\alpha b - a_1}{\beta a - a_3} \Big)^{1/({\gamma}-1)}$. Therefore, $\mathcal{S} \cap (0,\Big( \frac{a}{b} \frac{\alpha b - a_1}{\beta a - a_3} \Big)^{1/({\gamma}-1)}) $ is an empty set.
		Further, by the condition on constant $\rho$, one has
		$$
			\lim_{T \to \infty} \E \Big[e^{-\rho T} ((a_1-b\alpha) X_{T} +(-\beta a +a_3) X_{T}^{\gamma})  \Big]  = -\infty.
		$$
		We claim that $\mathcal{S} = [x^{\star} ,\infty)$. Otherwise, if there is a bounded open set $(o_1,o_2)$ belongs to $\mathcal{D}$, 
		notice that $g$ and $v$ satisfy \eqref{ineq:var} with boundary $v=g$ at points $o_1$ and $o_2$. 
		By the uniqueness of viscosity solution, one can deduce that $v=g$ which contradicts the assumption $(o_1,o_2) \subset \mathcal{D}$. 
		
		Next, let prove that $v(0+)=0$. First, we have 
		$$
			v(x) \geq g_v(x) \, \text{ for all } x>0,
		$$
		and hence  $v(0+) \geq g_v(0+)=0$.
		Further, it is obvious that $v(x) \leq 0 $ for all $x \in \mathbb{R}^+.$ Therefore, we can conclude $v(0+)=0.$
		
			

		According to the previous analysis, we know that the stopping region $ \mathcal{S}$ has the form $[x^{\star}, \infty)$ for some $x^{\star} >0$. From the Theorem 6 of Chapter 3 in \cite{Shiryaev}, the hitting time $\htau := \inf\{ t \ge 0 ~: X_t \ge x^{\star} \}$ is an optimal solution to the optimal stopping problem \eqref{eq:OT_B}.

		Since $v(x)$ is a viscosity solution of (\ref{ineq:var}) on $\mathbb{R}^+$ and $v(x)>g_v(x)$ on $(0, x^{\star})$, actually $v(x)$ is a viscosity solution of $\rho v - \mathcal{L} v = 0$ on $(0, x^{\star})$. 
		By the uniqueness of solution of linear PDE on a bounded domain, we claim that $v(x)$ is the solution in a classical sense. Therefore, 
		$v(x)$ satisfies 
		$$
		\rho v - \mu x v' - \frac{1}{2} \sigma^2 x^2 v''  = 0, \, 
		\text{in } (0,x^{\star}).
		$$   
		
		Notice that $v(0+) = 0 $. Upon solving the ODE, we have 
		$$
		v(x) = A x^m  ,
		$$
		where $m$ is defined as before and $A$ is some constant. Furthermore, the smooth-fit principle of optimal stopping problems guarantees $v(x)$ is $\mathcal{C}^1$ on $x^{\star},$ so we can obtain an equation system 
		\[
		\begin{cases}
			A (x^{\star})^m  & =  (a_1-b\alpha)  x^{\star} + (-\beta a +a_3)(x^{\star})^{\gamma}  , \\
			
			A m(x^{\star})^{m-1} & = (a_1-b\alpha)  +  {\gamma} (-\beta a +a_3) (x^{\star})^{{\gamma}-1} . \\
		\end{cases}
		\]

After solving the equation, we get
\begin{align*}
x^{\star} & = \Big( \frac{m-1}{{\gamma}-m} \frac{\alpha b - a_1}{\beta a-a_3 }  \Big)^{1/({\gamma}-1)} , \quad 
A  =  g_v(x^{\star}) \big( \frac{1}{x^{\star}} \big)^m.
\end{align*}

It follows that
$$
v(x) =  \big( (a_1-b\alpha)  x^{\star} + (-\beta a +a_3)(x^{\star})^{\gamma} \big) \big( \frac{x}{x^{\star}} \big)^m \1_{\{x \in (0, x^{\star})\}}  + \big( (a_1-b\alpha)  x + (-\beta a +a_3)(x)^{\gamma} \big) \1_{\{x \in [x^{\star}, \infty)\}} .
$$
Finally, using the definition of $x^{\ast}$, we conclude that 
$$u(x) =
\Big( \alpha bx + \beta a x^{\gamma} +  ( g(x^{\ast})  - \alpha bx^{\ast} - \beta a (x^{\ast})^{\gamma} ) \Big( \frac{x}{x^{\ast}}\Big)^m \Big) \1_{\{x \in (0, x^{\ast})\}}+ g(x) \1_{\{x \in [x^{\ast}, \infty)\}},$$
and the hitting time $\htau := \inf\{ t \ge 0 ~: X_t \ge x^{\ast} \}$ is an optimal solution to the optimal stopping problem \eqref{eq:OT_A}. \qed\end{proof}
	
	\begin{Remark}
	For any $\tau \in \mathcal{T}$, we take $\tau_n := \inf \{ t>0; X_t \geq n \} \wedge n \wedge \tau $. It is clear that $\tau_n \uparrow \tau$. According to Ito's formula, we get
	$$
		\E(e^{-\rho \tau_n} X_{\tau_n}^{\gamma}) = x_0^{\gamma} + \E (\int_0^{\tau_n} (-\rho +\gamma \mu + \frac{\gamma(\gamma-1)\sigma^2}{2} ) e^{-\rho t} X_t^{\gamma} dt ) + \E ( \int_0^{\tau_n} e^{-\rho t} \gamma \sigma X_t^{\gamma} dW_t).
	$$
	Then, taking the limit $n \to \infty$ and by the definition of $\mathcal{T}$,
	it follows that 
	$$\E \big[ e^{-\rho \tau} X_{\tau}^{\gamma} \big] = x_0^{\gamma} + \E \Big[ \int_0^{\tau} (-\rho +\gamma \mu + \frac{\gamma(\gamma-1)\sigma^2}{2} ) e^{-\rho t} X_t^{\gamma} dt \Big].$$
	In general, since $e^{-\rho t} X_t^{\gamma}$ is a submartingale, 
	by Theorem 5.4.4 of \cite{Durrett}, we has
	$$
		\E \Big[ \sup_{0\leq t \leq \tau} e^{-\rho t} X_t^{\gamma} \Big] 
		\leq 
		\frac{e}{e+1} \Big(1+\E\Big[ e^{-\rho \tau} X_{\tau}^{\gamma} \log^+( e^{-\rho \tau} X_{\tau}^{\gamma} ) \Big] \Big) < \infty.
	$$ 
	\end{Remark}

\begin{Remark} \label{Remark: hitting_sol}
	Once we know that the optimal stopping time is a hitting time or if we only consider hitting time solutions, we can derive the solution to the optimal problem (\ref{eq:OT_B}) easily from another perspective. Combining this with Proposition \ref{prop:payment}, the value function of (\ref{eq:OT_B}) is equivalent to
	$$
	v(x) = \sup_{y} ((a_1 - b \alpha) y + (-\beta a + a_3) y^{\gamma}) \left(\frac{x}{y}\right)^m,
	$$
	when considering $\tau = \inf\{t > 0: X_t = y\}$. 
\end{Remark}
	
	\begin{Proposition} \label{prop:verification}
		Let us consider the optimal stopping problem 
		$$
			u(x) : = \sup_{\tau \in \mathcal{T}} \E^{0,x} [ e^{-\rho \tau } g(X_{\tau}) ], ~~~x > 0,
		$$
		where $g(x)  \in C^1(\mathbb{R}^+) \cap C^2(\mathbb{R}^+ \setminus \mathcal{V})$ and $g(x)$ is a polynomial of degree less than or equal to $\gamma$, $\mathcal{V}$ is a finite set in $\mathbb{R}^+$. 
		Assume that the function $\varphi(x) $ satisfies the following properties: 
		\begin{itemize}
			\item[(1)] $\varphi \in C^1(\mathbb{R}^+) \cap C^2(\mathbb{R}^+ \setminus \mathcal{V})   $ and $\varphi$ is continuous at $0$
			\item[(2)] $\varphi''$ is locally bounded near each point of $\mathcal{V}$
			\item[(3)]  $\varphi(0+) = g(0+)$ and for some point $x_g$, $\varphi > g $ on $(0,x_g)$ and $\varphi = g$ on $[x_g,\infty)$
			\item[(4)] $\rho \varphi - \mathcal{L} \varphi \geq 0$ on $(x_g, \infty) \setminus \mathcal{V} $ and $\rho \varphi - \mathcal{L} \varphi = 0$ on $(0, x_g)$
			\item[(5)] $\tau_{x_g} : = \inf \{ t>0; X_t \notin (0, x_g) \} < \infty $ a.s. for all $ x \in \mathbb{R}^+$
			\item[(6)] the family $\{ \varphi(X_{\tau}) ; \tau \leq \tau_{x_g} \}$ is uniformly integrable for all $x \in \mathbb{R}^+$.
		\end{itemize}
		Then $\varphi(x) = u(x)$ and $\tau^{\ast} = \tau_{x_g}$ is an optimal stopping time.
		
	\end{Proposition}
	\begin{proof}
		According to the conditions (1) (2) and the Theorem D.1 of \cite{Oksendal}, we can obtain a sequence of functions $\varphi_j \in C^2(\mathbb{R}^+) $ such that 
		\begin{itemize}
			\item[(i)] $\varphi_j \to \varphi $ uniformly on compact subsets of $\mathbb{R}^+$, as $j \to \infty$
			\item[(ii)] $\rho \varphi_j - \mathcal{L} \varphi_j \to \rho \varphi - \mathcal{L} \varphi$ uniformly on compact subsets of $\mathbb{R}^+ \setminus \mathcal{V}$, as $j \to \infty$
			\item[(iii)] $ \{ \rho \varphi_j - \mathcal{L} \varphi_j \}_{j=1}^{\infty}$ is locally bounded on $\mathcal{V}$       
		\end{itemize}
		For any $R>0$, define $T_R = \inf\{ t>0; X_t \geq R \} \wedge R$ and $\tau \in {\mathcal{T}}$ be a stopping time. For all $x \in \mathbb{R}^+$, by Dynkin's formula, we obtain 
		$$
		\E^{0,x} \big[e^{-\rho (\tau \wedge T_R)} \varphi_j(X_{\tau \wedge T_R}) \big] 
		=
		\varphi_{j}(x) -\E^{0,x} \Big[\int_0^{\tau \wedge T_R} e^{-\rho s} (\rho \varphi_j - \mathcal{L} \varphi_j )(X_t)dt \Big].
		$$
		Taking the limit $j \to \infty$, one has
		$$
		\varphi(x) = \E^{0,x} \Big[ e^{-\rho (\tau \wedge T_R)} \varphi(X_{\tau \wedge T_R}) + \int_0^{\tau \wedge T_R} e^{-\rho s} (\rho \varphi - \mathcal{L}\varphi)(X_t) dt \Big].
		$$
		Based on our assumption that $\varphi(x) \geq g(x)$ and $\rho \varphi - \mathcal{L}\varphi \geq 0$, we will have
		$$
		\varphi(x) \geq E^{0,x}[e^{-\rho (\tau \wedge T_R)} g(X_{\tau \wedge T_R})].
		$$
		Take limitation for $R$ and apply the property of 
		$g(x)$ and $\widetilde{\mathcal{T}}$, we have 
		$$
		\varphi(x) \geq \lim_{R \to \infty} E^{0,x}[e^{- \rho (\tau \wedge T_R)} g(X_{\tau \wedge T_R})] \geq E^{0,x}[e^{- \rho \tau} g(X_{\tau })]  .
		$$
		
		Due to the arbitrariness of $\tau$, we have $\varphi(x) \geq u(x)$ for any $x \in \mathbb{R}^+$. Furthermore, if $x \in (x_g,\infty)$, we have known $\varphi(x) = g(x) \leq u(x)$, so $\varphi(x) = u(x)$ and $\hat \tau = 0.$ If $x \in (0,x_g)$, let $\tau_g:= \inf \{t >0, X_t \notin (0,x_g) \}$, by Dynkin's formula, for any $x \in (0,x_g)$, 
		$$
		\begin{array}{lll}
			\varphi(x) & = & \lim_{j \to \infty} \E^{0,x}[e^{-\rho (\tau_g \wedge T_R)}\varphi_j(X_{\tau_g \wedge T_R}) + \int_0^{\tau_g \wedge T_R} e^{-\rho s} (\rho \varphi_j - \mathcal{L}\varphi_j)(X_t) dt]. \\    
		\end{array}
		$$
		Utilizing dominated convergence theorem, we can have 
		$$
		\varphi(x) =  \E^{0,x}[e^{-\rho (\tau_g \wedge T_R)} \varphi(X_{\tau_g \wedge T_R}) + \int_0^{\tau_g \wedge T_R} e^{-\rho s} (\rho \varphi -\mathcal{L}\varphi)(X_t) dt].
		$$
		Since $\rho \varphi -\mathcal{L}\varphi = 0$ on $(0,x_g)$, we obtain
		$$
		\varphi(x) =  \E^{0,x}[e^{-\rho (\tau_g \wedge T_R)} \varphi(X_{\tau_g \wedge T_R}) ].
		$$
		Using the fact that $|X_{\tau_g \wedge T_R}|$ is uniformly bounded for all $R$, we will have 
		$$
		\varphi(x) =  \lim_{R \to\infty} \E^{0,x}[e^{-\rho (\tau_g \wedge T_R)} \varphi(X_{\tau_g \wedge T_R}) ] = \E^{0,x}[e^{-\rho \tau_g } \varphi(X_{\tau_g}) ] = \E^{0,x}[e^{-\rho \tau_g } g(X_{\tau_g}) ] \leq u(x).
		$$
		
		Therefore, $\varphi(x) = u(x)$ and $\hat \tau = \tau_g$ is optimal if $x \in (0,x_g)$. We can conclude that $\varphi(x) = u(x)$ and $\hat \tau = \tau_g$ is an optimal stopping time for all $x \in \mathbb{R}^+$.
	\qed
	\end{proof}

\subsection{Expected payments}
	
	Let $X_0 = x_0 > 0$ and 
	$$
		\htau :=  \mbox{inf}\{t\geq 0: X_t \geq  \hat x \},
		~~\mbox{for some}~
		\hat x > x_0,
	$$
	we compute the expectation values $\E[e^{- \rho \htau}]$ and $\E[\htau]$ below.

	\begin{Proposition} \label{prop:payment} One has
	\begin{align*}
		\E \big[ e^{-\rho \htau} \big] = \Big( \frac{x_0}{\hat x} \Big)^m, 
		\quad \mbox{and}~~
		\E[\htau] = \frac{-1}{\mu - \frac{\sigma^2}{2}} \log\Big(\frac{x_0}{\hat x}\Big),
	\end{align*}
	where $m =\frac{1}{2} - \frac{\mu}{\sigma^2} + \sqrt{ \Big( \frac{1}{2} - \frac{\mu}{\sigma^2} \Big)^2 +\frac{2\rho}{\sigma^2} }$.
	\end{Proposition}
	\begin{proof} Let 
	$\nu : = (\mu-\frac{1}{2} \sigma^2) \sigma^{-1}  \ge 0 $ , $\xi : = \sqrt{2 \rho + \nu^2}$ , $\theta := \frac{1}{\sigma} \log(\hat x/x_0)$.
	Then it follows by \cite{JM} that
	$$
		\P \big[ \htau < t \big]
		= 
		(\frac{-\theta+\nu t}{\sqrt{t}}) + e^{2\nu \theta} N(\frac{-\nu t -\theta}{\sqrt{t}})
	$$
	and hence $\P[ \htau = \infty ] =0$. 
	Further, it follows by subsection 3.3.3 of \cite{JM} that
	$$
		\E \big[ e^{-\rho \htau} \1_{ \{ \htau <t \}} \big] 
		=
		e^{(\nu - \xi) \theta} N(\frac{\xi t-\theta}{\sqrt{t}}) + e^{(\xi +\nu)\theta} N(\frac{-\xi t -\theta}{\sqrt{t}}),
	$$
	where $\nu = (\mu-\frac{1}{2} \sigma^2) \sigma^{-1} $ , $\xi = \sqrt{2 \rho + \nu^2}$ and $\theta = \frac{1}{\sigma} ln(\hat x/x_0)$.
	Let $t \to \infty$, and notice that $\frac{(\nu - \xi)}{\sigma} = -m$, 
	it follows that 
	$$
		\E \big[ e^{-\rho \htau} \big] = e^{(\nu - \xi) \theta} = e^{ \frac{(\nu - \xi)}{\sigma} \log (\hat x/x_0) }
		= 
		\Big( \frac{x_0}{\hat x} \Big)^m.
	$$
	
	We observed that $\frac{1}{2} - \frac{\mu}{\sigma^2} < 0$. Thus, we can express the square root as: 
		$$
		\sqrt{ (\frac{1}{2} - \frac{\mu}{\sigma^2})^2}  = \frac{\mu}{\sigma^2} - \frac{1}{2}. 
		$$
		Next, let us define the moment generating function of $\htau$ by 
		$$
			M_{\htau}(\rho) : = \E(e^{\rho \htau}) = (\frac{x}{\hat x})^{\frac{1}{2} - \frac{\mu}{\sigma^2} + \sqrt{ \Big( \frac{1}{2} - \frac{\mu}{\sigma^2} \Big)^2 -\frac{2\rho}{\sigma^2} }},
		$$
		one can then obtain $\E[\htau]$ by calculating the derivative of the moment generating function at $0$: 
		$$
			M_{\htau}'(0) =\frac{-1}{\mu - \frac{\sigma^2}{2}} \log \Big(\frac{x_0}{\hat x}\Big).
		$$ 
		\qed
	\end{proof}

\subsection{Proof of Theorem~\ref{theo:Nfirms}.}
\label{app:main1}

	$\mathrm{(i)}$	
	When $N=2$, since $\lambda_1 + \lambda_2 =1$, it follows that
	$$
		\hat  x_1 < \hat x_2
		\Longleftrightarrow
		\frac{2 \pi_1 \lambda_1}{1 - \lambda_2^{\gamma}} < \frac{\pi_2 (\lambda_2+1) -\pi_1 \lambda_1}{\lambda_2^{\gamma}}
		\Longleftrightarrow
		\frac{\pi_1}{\pi_2} < \frac{(1+\lambda_2)(1-\lambda_2^{\gamma})}{(1-\lambda_2)(1+\lambda_2^{\gamma})},
	$$
	which is always true as $\lambda_2 \in (0,1)$ and $\pi_1 \leq \pi_2$.
	
	When $N > 2$, by the \eqref{def:hitting_point},
	\begin{equation} \label{eq:order_condition_inter}
		\hat{x}_{i} < \hat{x}_{i+1}
		\Longleftrightarrow
		\frac{\pi_i(\lambda_i + \tilde{\lambda}_i) - \pi_{i-1} \tilde{\lambda}_{i-1}}{\bar{\lambda}_i^{\gamma} - \bar{\lambda}_{i+1}^{\gamma}} < \frac{\pi_{i+1}(\lambda_{i+1} + \tilde{\lambda}_{i+1}) - \pi_{i} \tilde{\lambda}_{i}}{\bar{\lambda}_{i+1}^{\gamma} - \bar{\lambda}_{i+2}^{\gamma}}.
	\end{equation}
        Since \( \bar{\lambda}_i> \bar{\lambda}_{i+1} \) and $\pi_i <\pi_{i+1}$ for any \( i \), we can conclude that both the numerator and denominator in \eqref{eq:order_condition_inter} are positive. 
        By direct computation, \eqref{eq:order_condition_inter} is equivalent to
        \begin{align}  \label{ineq:order_condition}
		&\big( \pi_i(\lambda_i + \tilde{\lambda}_i) - \pi_{i-1} \tilde{\lambda}_{i-1} + \pi_{i+1}(\lambda_{i+1} + \tilde{\lambda}_{i+1}) - \pi_{i} \tilde{\lambda}_{i} \big) \bar \lambda_{i+1}^\gamma \nonumber \\
		<~ &
		\big( \pi_{i+1}(\lambda_{i+1} + \tilde{\lambda}_{i+1}) - \pi_{i} \tilde{\lambda}_{i} \big)\bar \lambda_i^\gamma   +\big( \pi_i(\lambda_i + \tilde{\lambda}_i) - \pi_{i-1} \tilde{\lambda}_{i-1} \big) \bar \lambda_{i+2}^\gamma 
        \end{align}
        Notice that the function \( x \mapsto x^{\gamma} \) is convex with $\gamma \ge 2$, then $\bar{\lambda}_i> \bar{\lambda}_{i+1} > \bar{\lambda}_{i+2}$ for any $i$ and hence 
	\begin{equation*} 
		\bar{\lambda}_{i+1}^{\gamma} < p \bar{\lambda}_{i}^{\gamma} 
		+ (1-p)\bar{\lambda}_{i+2}^{\gamma}, \text{ for any } p \in [0,1].
	\end{equation*}
        This is enough to conclude that \eqref{ineq:order_condition} always holds and hence $\hat{x}_{i} < \hat{x}_{i+1}$ is always true as claimed in $\mathrm{(i)}$.
	
	\vspace{0.5em}

	\noindent $\mathrm{(ii)}$ To prove $\mathrm{(ii)-(iv)}$, it is enough to solve the multiple optimal stopping problem \eqref{eq:regu_equiv}.
	For each $i=1, \cdots, N$, let us introduce the optimal stopping problem:
	$$
		v_i (x) 
		~:=~
		\sup_{\tau \in \Tc} \E \Big[ 
			\int_0^{\tau} e^{-\rho t} \big( \bar \pi_i X_t - \ell (1- \tilde \lambda_{i-1})^{\gamma} X_t^{\gamma} \big) dt 
			-
			e^{-\rho \tau} g_i(X_{\tau})
		\Big| X_0 = x
		\Big]
	$$
	with $v_{N+1}(\cdot) \equiv 0$ and
	$$
		g_i(x) := \big( \tilde \lambda_i c_i - \tilde \lambda_{i-1} c_{i-1} \big) x + v_{i+1}( x).
	$$
	We claim that, for each $i=1, \cdots, N$,
	\begin{equation} \label{eq:claim_vi}
		v_i(x) =  \Big( \bar \pi_i b x + (1 - \tilde \lambda_{i-1})^{\gamma} a \ell x^{\gamma} + \Delta_i(\hat x_i) x^m \Big) \1_{\{x \in (0, \hat x_i)\}} + g_i(x) \1_{\{ x \in [\hat x_i, \infty)\}},
	\end{equation}
	where
	$$
		\Delta_i(x) := \big( g_i(x) - \bar \pi_i b x - (1 - \lambda_{i-1})^{\gamma} a \ell x^{\gamma} \big) x^{-m}.
	$$
	
	We next prove Claim \eqref{eq:claim_vi} by induction arguments in the following.
	First, for $i=N$, one can directly apply Proposition \ref{prop:OT_solution} to obtain that
	$$
		v_N (x) =
		\Big( \bar{\pi}_N bx + (1 - \tilde \lambda_{N})^{\gamma} a \ell x^{\gamma} +\Delta_N (\hat x_N) x^m \Big) \1_{\{x \in (0, \hat x_N)\}}
		+ 
		g_N(x) \1_{\{x \in [ \hat x_N, \infty)\}},
	$$
	so that Claim \eqref{eq:claim_vi} holds for $i=N$.
	Moreover, the optimal stopping time for the optimal stopping problem of $v_N$ is given by $\htau_N= \inf\{ t \ge 0 ~: X_t \ge \hat x_N \}$.
	
	Next, for $i= N-1$, we study the optimal stopping problem:
	\begin{equation} {\label{eq:N-1}}
		v_{N-1}(x)
		~:=~ 
		\sup_{\tau \in \mathcal{T}} 
		\E \Big[ 
			\int_0^{\tau} e^{-\rho t} f_{N-1}(X_t) \dd t + e^{-\rho \tau} g_{N-1}(X_{\tau})
		\Big| X_0 = x
		\Big],
	\end{equation}   
	where $f_{N-1} (x) =\bar \pi_{N-1} x-\ell (1- \tilde \lambda_{N-N})^{\gamma} x^{\gamma}$ and 
	\begin{align*}
		g_{N-1}(x)
		&~=~
		 \Big( - (\tilde{\lambda}_{N-1} c_{N-1} - \tilde{\lambda}_{N-2} c_{N-2})x+  \bar{\pi}_N bx + \bar{\lambda}_{N}^{\gamma} a \ell x^{\gamma} +\Delta_N(\hat x_N) x^m \Big) \1_{\{x \in (0, \hat x_N)\}} \\
		 &~ -~ (\tilde{\lambda}_N c_N - \tilde{\lambda}_{N-2} c_{N-2})x  \1_{\{x \in [ \hat x_N, \infty)\}} .
	\end{align*}
	As in the proof of Proposition \ref{prop:OT_solution}, by applying It\^o's formula on $X^{\gamma}$, one obtains that
	$$
		v_{N-1}(x)
		=
		\tilde v_{N-1}(x) + b\bar \pi_{N-1} x + a \ell \big(1- \tilde \lambda_{N-2} \big)^{\gamma} x^{\gamma}  ,
	$$
	where
	$$
		\tilde v_{N-1}(x) 
		~:=~
		\sup_{\tau \in \mathcal{T}} \E \Big [ e^{-\rho \tau} \tilde g (X_{\tau})  \Big| X_0 = x \Big],
	$$
	for 
	$ \tilde g (x) : =   \tilde g_1(x) \1_{\{x \in (0, \hat x_N)\}} + \tilde g_2(x)  \1_{\{x \in [ \hat x_N, \infty)\}}$,
	$$
		\tilde g_1(x) := - (\tilde{\lambda}_{N-1} c_{N-1} - \tilde{\lambda}_{N-2} c_{N-2})x -\lambda_{N-1}c_{N-1} x -a \ell((1- \tilde \lambda_{N-2})^{\gamma} - (1- \tilde \lambda_{N-1})^{\gamma})x^{\gamma}  +\Delta_N(\hat x_N) x^m,
	$$
	and
	$$
		\tilde g_2(x) :=   -(\tilde{\lambda}_N c_N - \tilde{\lambda}_{N-2} c_{N-2})x -b\bar \pi_{N-1} x - a \ell (1- \tilde \lambda_{N-1})^{\gamma} x^{\gamma} .
	$$
	Further, let us define	
	$$
		\varphi_{N-1}(x) : = \Big(\Delta_{N-1}(\hat  x_{N-1}) x^m \Big) \1_{\{x \in (0,  \hat x_{N-1})\}} + g'(x) \1_{\{x \in [\hat x_{N-1}, \infty)\}} .
	$$
	By applying Proposition \ref{prop:verification}, one can obtain that
	$$
		\tilde v_{N-1}(x) = \varphi_{N-1}(x)  = \E \Big[ e^{-\rho \hat \tau_{N-1}} \tilde g(X_{\hat \tau_{N-1} }) \Big| X_0 = x  \Big]. 
	$$
	Indeed, to check the conditions in Proposition \ref{prop:verification}, one can set $\mathcal{V} := \{\hat x_{N-1}, \hat x_{N} \}  $ and $x_g = \hat x_{N-1}$, 
	then Conditions (1)(2)(5)(6) in Proposition \ref{prop:verification} is easy to verify.
	For Condition (3) in Proposition \ref{prop:verification},
	let us define $h(x) : = \varphi_{N-1}(x) - \tilde g(x)$, so that that $h(x) = 0 $ when $x \in [\hat x_{N-1}, \infty)$. 
	Further, for $x \in (0,  \hat x_{N-1})$, one has
	\begin{align*}
		h(x) 
		~=~& a \ell((1- {\lambda}_{N-2})^{\gamma} - (1- {\lambda}_{N-1})^{\gamma})x^{\gamma} + (\Delta_{N-1}(\hat x_{N-1})-\Delta_N(\hat x_N) ) x^m \\
		&~+~ (\tilde{\lambda}_{N-1} c_{N-1} +\lambda_{N-1}c_{N-1} - \tilde{\lambda}_{N-2} c_{N-2})x.
	\end{align*}
	Since $\lambda_{N-1} c_{N-1} +\tilde{\lambda}_{N-1} c_{N-1} - \tilde{\lambda}_{N-2} c_{N-2}>0$, recall the definition of $\Delta_i(\cdot)$, one can compute that
	$$
		h(0+)=0,\, h(\hat x_{N-1}) = 0, \, h'(0+)=0, \, h'(\hat x_{N-1}) = 0,\, h''(\hat x_{N-1}) > 0.
	$$
	Thus $h(x)>0$ for all $x \in (0,\bar x_{N-1})$ and hence Condition (3) holds true. 
	Notice that Condition (4) is easy to check as  $\rho \varphi_{N-1} - \mathcal{L} \varphi_{N-1} =0$ on $\mathbb{R}^+$.
	Therefore,
	$$
		v_{N-1}(x)
		=
		 \varphi_{N-1}(x) + b\bar \pi_{N-1} x + a \ell \big(1- \tilde \lambda_{N-2} \big)^{\gamma} x^{\gamma} ,
	$$
	and 
	$\hat \tau_{N-1}$ is an optimal stopping time for the optimal stopping problem $v_{N-1}(x)$.
	By iterating the same argument for $i - N-2, \, N-3, \dots, 1$, 
	one can conclude that the value function $v(x_0) = v_1(x_0)$ is given by  \eqref{eq:value_fun}, and that $\hat \tau_1 < \cdots < \hat \tau_N$ is an optimal solution to \eqref{eq:regu_equiv}.
	This proves $\mathrm{(iv)}$.

	\vspace{0.5em}
	
	Finally, following  \cite{He23}, the optimal incentive scheme $\widehat Y$ should be designed such that, for $i=1, \cdots, N$, 
	\begin{align*}
		\widehat Y_{\htau_i} 
		~=~
		\E \bigg[  \sum_{j=i}^{N-1} \int_{\htau_{j}}^{\htau_{j+1}}  e^{-\rho (s -\htau_i) } \pi_j X_t ds + \int_{\htau_N}^\infty e^{-\rho (s -\htau_i) } \pi_N X_t ds \bigg| \mathcal{F}_{\htau_i} \bigg] 
		~=~
		c_i X_{\htau_i} +d_i X_{\htau_i}^m,
	\end{align*} 
	with
	$$
		c_i := \frac{\pi_i}{\rho - \mu},  
		~~~
		d_i := \sum_{j=i+1}^N \frac{\pi_j - \pi_{j-1}}{\rho - \mu} \Big(\1_{ \{x_0 \geq \hat x_j\}}  x_0^{1-m} + \1_{ \{x_0 < \hat x_j\}} \hat x_j^{1-m} \Big),
	$$
	and
	$$
		\widehat Y_t < c_i X_t + d_i X^m_t,
		~~\mbox{for}~~
		t \in (\hat \tau_{i-1}, \hat \tau_i).
	$$
	Moreover, under the optimal scheme $\widehat Y$, the optimal exit time of Agent $i$ is given by $\hat \tau_i$, for $i=1, \cdots, N$.
	We hence conclude the proof of $\mathrm{(ii)}$ and $\mathrm{(iii)}$.
	\qed

\subsection{Proof of results on the two interacted markets}

	The prove Theorem \ref{thm:NashE}, let us consider the optimal stopping problem of one regulator.

	\begin{Proposition} \label{prop:opt_stopping_1}
		Assume that the initial condition $X_0 = x_0$ satisfies $x_0 < x_{1, l}$, and that $\tau_2$ is given by $\tau_2 :=\inf \{ t>0~: X_t \geq x_2\} $ with some constant $x_2 > x_0$.
		Then
		\begin{itemize}
			\item $\tau_{1,l} := \inf \{ t > 0 ~: X_t \ge x_{1,l} \}$ is the optimal stopping time to the problem \eqref{eq:pb_regulator1}  if $x_2 \ge x_1$;
			
			\item $\tau_{1,h} :=  \inf \{ t > 0 ~: X_t \ge x_{1,h} \}$ is the optimal solution to \eqref{eq:pb_regulator1} if $x_2 \le x_1$.
		\end{itemize}
		
	\end{Proposition}
	\begin{proof}
	By applying It\^o's formula on $(X^{\gamma}_t)_{t \ge 0}$, one can reformulation the optimal stopping problem \eqref{eq:pb_regulator1} as
	\begin{align*} 
		w_1(x_0) ~:=~ \pi_1 \lambda_1 b x_0 & +\frac{1}{2} a \ell x_0^{\gamma} \nonumber \\
			+~\sup_{\tau_1 \in \Tc} 
			\E \Big[ & \1_{\tau_1 \geq \tau_2} \Big(e^{-\rho \tau_1} \big(-2\pi_1 \lambda_1 b X_{\tau_1} -\frac{1}{2} a \ell \lambda_1^{\gamma} X_{\tau_1}^{\gamma} \big) - \frac{1}{2} a\ell (1-\lambda_1^{\gamma}) x_2^{\gamma} \big( \frac{x_0}{x_2}\big)^m \Big) \nonumber \\
			& + \1_{\tau_1 < \tau_2} \Big(e^{-\rho \tau_1} \big(-2\pi_1 \lambda_1 b X_{\tau_1} -\frac{1}{2} a \ell (1-\lambda_2^{\gamma}) X_{\tau_1}^{\gamma} \big) - \frac{1}{2}  a\ell \lambda_2^{\gamma} x_2^{\gamma} \big( \frac{x_0}{x_2}\big)^m \Big) \Big] .
	\end{align*}
	Notice that $x_0 < x_2$, by considering the two sub-problems, one has 
	$$
		w_1 (x_0) = \max( w_{1,1}(x_0), w_{1,2} (x_0) ),
	$$ 
	where	
	\begin{align*}
		w_{1,1}(x_0) :=&
		\sup_{\tau_1 \in\Tc, ~\tau_1 \geq \tau_2} \E \Big[ e^{-\rho \tau_1} \big( -2\pi_1 \lambda_1 b X_{\tau_1} - \frac{1}{2}  a \ell \lambda_1^{\gamma} X_{\tau_1}^{\gamma} \big) \Big] 
			+\pi_1 \lambda_1 b x_0 +\frac{1}{2} a \ell x_0^{\gamma} - \frac{1}{2} a\ell (1-\lambda_1^{\gamma}) x_2^{\gamma} \big( \frac{x_0}{x_2}\big)^m    \\
		=&~ \1_{x_{1,h} <x_2} \big( -2 \pi_1 \lambda_1 b x_2 - \frac{1}{2} a\ell \lambda_1^{\gamma} x_2^{\gamma} \big) \big( \frac{x_0}{x_2} \big)^m \\
			&+ \1_{x_{1,h} \geq x_2} \big(-2 \pi_1 \lambda_1 b x_{1,h}-\frac{1}{2} a \ell \lambda_1^{\gamma} x_{1,h}^{\gamma} \big) \big(\frac{x_0}{x_{1,h}} \big)^m 
			~+~ 
		\pi_1 \lambda_1 b x_0 +\frac{1}{2} a \ell x_0^{\gamma} - \frac{1}{2} a\ell (1-\lambda_1^{\gamma})  x_2^{\gamma} \big( \frac{x_0}{x_2}\big)^m ,
	\end{align*}
	with optimal solution given by $ \tau_{1,h} := \inf \{ t>0 , X_t \geq x_{1,h} \} $ if $x_{1,h} \geq x_2$ or $\tau_2$ if $x_{1,h} < x_2$.  
	The second equality is based on the technique mentioned in Remark \ref{Remark: hitting_sol}.  
	Similarly, $w_{1,2}(x_0)$ is the value function of the second sub-problem:
	\begin{align*}                          
		w_{1,2}(x_0) :=& 
		\sup_{\tau_1 \in \Tc, ~\tau_1 < \tau_2} \E \Big[ e^{-\rho \tau_1}(-2\pi_1 \lambda_1 b X_{\tau_1} -\frac{1}{2} a \ell (1- \lambda_2^{\gamma}) X_{\tau_1}^{\gamma} ) \Big] 
			+\pi_1 \lambda_1 b x_0 +\frac{1}{2} a \ell x_0^{\gamma} - \frac{1}{2} a\ell \lambda_2^{\gamma} x_2^{\gamma} \big( \frac{x_0}{x_2}\big)^m    \\
		= &~ \1_{x_{1,l} <x_2} \big(-2 \pi_1 \lambda_1 b x_{1,l} - \frac{1}{2} a\ell (1-\lambda_2^{\gamma}) x_{1,l}^{\gamma} \big) \big( \frac{x_0}{x_{1,l}} \big)^m \\
			&+ \1_{x_{1,l} \geq x_2} \big(-2 \pi_1 \lambda_1 b x_{2} -\frac{1}{2} a \ell (1-\lambda_2^{\gamma}) x_{2}^{\gamma}\big ) \big(\frac{x_0}{x_{2}} \big)^m
			~+~
			\pi_1 \lambda_1 b x_0 +\frac{1}{2} a \ell x_0^{\gamma} - \frac{1}{2}  a\ell \lambda_2^{\gamma} x_2^{\gamma} \big( \frac{x_0}{x_2}\big)^m ,
	\end{align*}
	with optimal stopping time given by $\tau_{1,l} := \inf \{ t>0 , X_t \geq x_{1,l} \}$ if $x_{1,l} < x_2$ and $\tau_2$ if $x_{1,l} \geq x_2$. 
	By comparing $w_{1,1}(x_0)$ and $w_{1,2}(x_0)$ for different positions of $x_2$ and use the constant $z_1$ defined in \eqref{eq:def_z},
	it follows that 
	$$
		w_1(x_0) = 
		\begin{cases}
			(-2 \pi_1 \lambda_1 b x_{1,h}- \frac{1}{2}  a \ell \lambda_1^{\gamma} x_{1,h}^{\gamma}) \big(\frac{x_0}{x_{1,h}} \big)^m +\pi_1 \lambda_1 b x_0 +\frac{1}{2}  a \ell x_0^{\gamma} - \frac{1}{2} a\ell (1-\lambda_1^{\gamma}) x_2^{\gamma} \big( \frac{x_0}{x_2}\big)^m,
			& \textrm{if $x_2 \leq z_1$}, \\
			(-2 \pi_1 \lambda_1 b x_{1,l} - \frac{1}{2} a\ell (1-\lambda_2^{\gamma}) x_{1,l}^{\gamma}) 
			\big(\frac{x_0}{x_{1,l}} \big)^m +\pi_1 \lambda_1 b x_0 + \frac{1}{2} a \ell x_0^{\gamma} - \frac{1}{2}  a\ell \lambda_2^{\gamma} x_2^{\gamma} \big( \frac{x_0}{x_2}\big)^m,
			& \textrm{if $x_2 > z_1$.}\\
		\end{cases}
	$$
	Moreover,  $\tau_{1,l}$ is the optimal solution to \eqref{eq:pb_regulator1} if $x_2 \ge z_1$, and $\tau_{1,h}$ is the optimal solution if  $x_2 \leq z_1$.
	\qed
	\end{proof}

\hs

Similarly, one can solve the optimal stopping problem \eqref{eq:pb_regulator2} of Regulator 2.
Recall that the constants $z_{2,l}$, $z_{2,h}$ and $z_2$ are defined in \eqref{eq:def_zlh} and \eqref{eq:def_z}, and one has $z_{2,l} < z_2 < z_{2, h}$.
Further, the optimal stopping problem \eqref{eq:pb_regulator2} can be reformulated to
	\begin{align*}
		w_2 (x_0) :=
		\pi_2 \lambda_2 b x_0 + & \frac{1}{2} a \ell x_0^{\gamma} \\
		+
		\sup_{\tau_2 \in \Tc} 
			\E \Big[ & \1_{\tau_1 \geq \tau_2} \Big(e^{-\rho \tau_2}(-2\pi_2 \lambda_2 b X_{\tau_2} -\frac{1}{2} a \ell (1-\lambda_1^{\gamma}) X_{\tau_2}^{\gamma} ) - \frac{1}{2} a\ell \lambda_1^{\gamma} x_1^{\gamma} \big( \frac{x_0}{x_1}\big)^m \Big) \\
			& + \1_{\tau_1 < \tau_2} \Big(e^{-\rho \tau_2}(-2\pi_2 \lambda_2 b X_{\tau_2} -\frac{1}{2} a \ell \lambda_2^{\gamma} X_{\tau_2}^{\gamma} ) - \frac{1}{2}  a\ell (1-\lambda_2^{\gamma}) x_1^{\gamma} \big( \frac{x_0}{x_1}\big)^m \Big) \Big] .
	\end{align*}
	
	\begin{Proposition} \label{prop:opt_stopping_2}
		Assume that $x_0 < x_1$, $x_0 < x_{2,l}$ and let $\tau_1 := \inf \{ t>0: X_t \geq x_1 \}$.
		Then the value function $w_2(x_0)$ of \eqref{eq:pb_regulator2} is given by
		\begin{align*}
			w_2 (x_0) &=~ \max( w_{2,1}(x_0) , w_{2,2}(x_0) ) \\
			&= 
			\begin{cases}
			(-2 \pi_2 \lambda_2 b x_{2,h}- \frac{1}{2} a \ell \lambda_2^{\gamma} x_{2,h}^{\gamma}) \big(\frac{x_0}{x_{2,h}} \big)^m +\pi_2 \lambda_2 b x_0 +\frac{1}{2} a \ell x_0^{\gamma} - \frac{1}{2} a\ell (1-\lambda_2^{\gamma}) x_1^{\gamma} \big( \frac{x_0}{x_1}\big)^m,
			& \textrm{if $x_1 \leq z_2$,} \\
			(-2 \pi_2 \lambda_2 b x_{2,l} - \frac{1}{2} a\ell (1-\lambda_1^{\gamma}) x_{2,l}^{\gamma}) 
			\big(\frac{x_0}{x_{2,l}} \big)^m +\pi_2 \lambda_2 b x_0 +\frac{1}{2} a \ell x_0^{\gamma} - \frac{1}{2} a\ell \lambda_1^{\gamma} x_1^{\gamma} \big( \frac{x_0}{x_1}\big)^m,
			& \textrm{if $x_1 > z_2$,}
			\end{cases}
		\end{align*}
		where
		\begin{align*}
			w_{2,1}(x_0) 
			=& 
			\sup_{\tau_2 \in \Tc, \tau_2 \le \tau_1} \E \Big[ e^{-\rho \tau_2} \big( -2\pi_2 \lambda_2 b X_{\tau_2} -\frac{1}{2} a \ell (1-\lambda_1^{\gamma}) X_{\tau_2}^{\gamma} \big) \Big] 
				+ \pi_2 \lambda_2 b x_0 +\frac{1}{2} a \ell x_0^{\gamma} - \frac{1}{2} a\ell \lambda_1^{\gamma} x_1^{\gamma} \big( \frac{x_0}{x_1}\big)^m    \\
			=&~
			\1_{x_{2,l} <x_1} (-2 \pi_2 \lambda_2 b x_{2,l} - \frac{1}{2} a\ell (1-\lambda_1^{\gamma}) x_{2,l}^{\gamma}) \big( \frac{x_0}{x_{2,l}} \big)^m \\
			&+~ \1_{x_{2,l} \geq x_1} (-2 \pi_2 \lambda_2 b x_{1}-\frac{1}{2} a \ell (1-\lambda_1^{\gamma}) x_{1}^{\gamma}) \big(\frac{x_0}{x_{1}} \big)^m
			~+~\pi_2 \lambda_2 b x_0 +\frac{1}{2} a \ell x_0^{\gamma} - \frac{1}{2} a\ell \lambda_1^{\gamma} x_1^{\gamma} \big( \frac{x_0}{x_1}\big)^m ,
		\end{align*}	
		and
		\begin{align*}
			w_{2,2}(x_0) 
			=&
			\sup_{\tau_1 < \tau_2} \E\Big [ e^{-\rho \tau_2}(-2\pi_2 \lambda_2 b X_{\tau_2} -\frac{1}{2} a \ell \lambda_2^{\gamma} X_{\tau_2}^{\gamma} ) \Big] 
				+\pi_2 \lambda_2 b x_0 +\frac{1}{2} a \ell x_0^{\gamma} - \frac{1}{2} a\ell (1-\lambda_2^{\gamma}) x_1^{\gamma} \big( \frac{x_0}{x_1}\big)^m    \\
			=&~
			\1_{x_{2,h} <x_1} (-2 \pi_2 \lambda_2 b x_{1} - \frac{1}{2} a\ell \lambda_2^{\gamma} x_{1}^{\gamma}) \big( \frac{x_0}{x_{1}} \big)^m \\
			& +~ \1_{x_{2,h} \geq x_1} (-2 \pi_2 \lambda_2 b x_{2,h}
			~-~ \frac{1}{2} a \ell \lambda_2^{\gamma} x_{2,h}^{\gamma}) \big(\frac{x_0}{x_{2,h}} \big)^m  
			+\pi_2 \lambda_2 b x_0 +\frac{1}{2} a \ell x_0^{\gamma} - \frac{1}{2}  a\ell (1-\lambda_2^{\gamma}) x_1^{\gamma} \big( \frac{x_0}{x_1}\big)^m.
		\end{align*}
		Moreover, $\tau_{2,h} := \inf\{ t > 0 ~: X_t \ge x_{2,h} \}$ is an optimal solution to \eqref{eq:pb_regulator2} if $x_1 \le z_2$,
		and $\tau_{2,l} := \inf\{ t > 0 ~: X_t \ge x_{2,l} \}$ is an optimal solution to \eqref{eq:pb_regulator2} if $x_1 \ge z_2$.
	\end{Proposition}
	
\hs

	\noindent {\bf Proof of Theorem \ref{thm:NashE}}
	It is an immediate consequence of Propositions \ref{prop:opt_stopping_1} and \ref{prop:opt_stopping_2}.
	\qed

\hs
	
	\noindent {\bf Proof of Proposition \ref{prop:indiv_contract}}
	$\mathrm{(i)}$ Let us define optimal stopping problem
	\begin{equation}
	\begin{split}
	\sup_{\tau_1 \in \mathcal{T}} 	K_1(x, \tau_1, \tau_2) &:= J_1(x, \tau_1, \tau_2) - \frac{1}{2} D(x, \tau_1, \tau_2), \\
	\sup_{\tau_2 \in \mathcal{T}} 	K_2(x, \tau_1, \tau_2) &:= J_2(x, \tau_1, \tau_2) - \frac{1}{2} D(x, \tau_1, \tau_2).
	\end{split}
	\label{eq:K_12}
	\end{equation}
	
	We observe that the optimal solution to the optimal social planner problem \eqref{eq:individual_compen} provides a Nash equilibrium solution in the sense of Definition \eqref{def:Nash}. 
	Indeed, if $(\tau_1^*, \tau_2^*)$ is optimal for the problem \eqref{eq:individual_compen}, it is clear that $\tau_1^*$ is optimal for the sub-problem $\sup_{\tau_1 \in \mathcal{T}} J(x_0, \tau_1, \tau^*_2) + \frac{1}{2} D(x_0, \tau_1, \tau^*_2)$, and $\tau_2^*$ is optimal for the sub-problem $\sup_{\tau_2 \in \mathcal{T}} J(x_0, \tau^*_2, \tau_1) + \frac{1}{2} D(x_0, \tau^*_2, \tau_1)$. 
	Therefore, it is a Nash equilibrium solution in the sense of Definition \eqref{def:Nash}.
			
	\vspace{0.5em}

	\noindent $\mathrm{(ii)}$ We next follow exactly the same approach in the proof of Theorem \ref{thm:NashE} to find all Nash equilibriums in the form of hitting times of $X$.
Concretely, 

\begin{itemize}

	\item When $\hat z_1 < \hat z_2$, 
	there are exactly two (hitting time) Nash equilibrium solutions $({\hat \tau_{1,h}}, {\hat \tau_{2,l}})$ and $ ( {\hat \tau_{1,l}} , {\hat\tau_{2,h}})$.
	By comparing the optimal reward, it turns out that $ ( {\hat \tau_{1,l}} , {\hat \tau_{2,h}})$ is the optimal solution to \eqref{eq:individual_compen} when one considers only hitting times.
	
	\item When ${\hat z_2} < {\hat z_1}$,
	one can apply the similar argument to conclude that $ ( {\hat \tau_{1,h}} , {\ \tau_{2,l}})$ is the optimal solution to \eqref{eq:individual_compen} when one considers only hitting times.
	
	\item When $\hat z_1 = \hat z_2$, there are exactly two (hitting time) Nash equilibrium solutions 
	given by $({\hat \tau_{1,h}}, {\hat \tau_{2,l}})$ and $ ( {\hat \tau_{1,l} },  {\hat \tau_{2,h}})$.
	As the two solutions have equal value functions $\hat u_1(x_0) = \hat u_2(x_0)$,
	both are optimal solution to \eqref{eq:individual_compen}.
\end{itemize}
\qed

\bibliographystyle{plain}

\end{document}